\documentclass[]{interact}

\usepackage{lineno}
\usepackage{xcolor}
\usepackage{pifont}%
\newcommand{\cmark}{\ding{51}}%
\newcommand{\xmark}{\ding{55}}%
\usepackage{amsmath,amssymb,amsfonts}
\usepackage{algorithmic}
\usepackage[flushleft]{threeparttable}
\usepackage{comment}
\usepackage{enumerate}
\usepackage{mathtools}
\usepackage{chngcntr}
\usepackage{amsthm}
\usepackage{graphicx}
\usepackage[colorlinks=true, citecolor=green, linkcolor=blue, urlcolor=blue]{hyperref}
\usepackage{textcomp}
\usepackage{xcolor}
\usepackage{rotating}
\usepackage{float}
\usepackage{algorithm2e}
\usepackage[utf8]{inputenc}
\usepackage[english]{babel}
\newtheorem{thm}{Theorem}
\DeclareMathAlphabet\mathbfcal{OMS}{cmsy}{b}{n}
\renewcommand\thesubsection{\Alph{subsection}}

\DeclareMathAlphabet\mathbfcal{OMS}{cmsy}{b}{n}
\modulolinenumbers[5]
\usepackage{apacite}
\let\cite\shortcite
\begin{document}

%\begin{frontmatter}

\title{LDA-2IoT : A Level Dependent Authentication using Two Factor for IoT Paradigm}

\author{
\name{Chintan Patel*, Nishant Doshi}
\affil{Department of Computer Science and Engineering,\\ Pandit Deendayal Energy University, Gandhinagar,India}
\email{*chintan.p592@gmail.com}
}

\maketitle
\begin{abstract}
The widespread expansion of the IoT based services are changing people's living habits. With the vast data generation and intelligent decision support system, an IoT is supporting many industries to improve their products and services. The major challenge for IoT developers is to design a secure data transmission system and a trustworthy inter-device and user-device communication system. The data starts its journey from the sensing devices and reaches the user dashboard through different medium. Authentication between two IoT devices provides a reliable and lightweight key generation system. In this paper, we put forward a novel authentication approach for the IoT paradigm. We postulate an ECC based two factor Level-Dependent Authentication for Generic IoT (LDA-2IoT) in which users at a particular level in the hierarchy can access the sensors deployed at below or the equal level of the hierarchy. We impart the security analysis for the proposed LDA-2IoT based on the Dolev-Yao channel and widely accepted random oracle based ROR model. We provide the implementation of the proposed scheme using the MQTT protocol. Finally, we set forth a performance analysis for the proposed LDA-2IoT system by comparing it with the other existing schemes. 
\end{abstract}

%\begin{keyword}
%Generic IoT \sep Home Area Network \sep ECC \sep Authentication \sep BAN Logic \sep AVISPA \sep ROR Model \sep MQTT
%\end{keyword}
%\end{frontmatter}
\section{INTRODUCTION}
\label{Sec:Introduction}
\noindent Internet of Things (IoT) is a network of an interconnected network of sensing devices, mobile and laptop users, routing devices, servers, and other computing devices with the communicating capabilities. An IoT connects billions of resource-constraint devices with the physical world using a  lightweight communication and security mechanism. An objective of the IoT system is to provide "any service" to "any user" on "any time" at "anywhere". So overall, the IoT is an integration of all the Cyber-Physical Systems (CPS) through the internet.

The generic IoT model discussed in  \cite{patel2018internet} shows an integration of the various IoT components and entities. In general, IoT users connect with the different IoT applications through the internet. The IoT applications like smart home, smart factory, smart transportation, smart agriculture, smart health communicate integrated data with the users through intermediate internet devices such as gateways and switches. 

Most of the devices deployed in IoT networks are sensing devices. The sensing devices are tiny devices capable of detecting the surrounding biological environment as well as the nonbiological environment. The sensing devices communicate data in the short-range area and transmit sensed data to the nearest gateway device through technologies such as Bluetooth, RFID, Zigbee, and WiFi. The sensing devices are resource constraint devices in terms of storage cost, power utilization, and computation capabilities. The traditional security mechanisms of the internet use complex cryptography mathematical operations. These operations require ample storage space and high computation memory. Thus, due to the availability of numerous resource constraint devices, it is indispensable to prototype a  lightweight security mechanism for end-to-end data communication in IoT Model. The proposed lightweight security mechanism must be efficient in terms of the computation capabilities, optimized in terms of memory and time utilization, and robust against the traditional and non-traditional security attacks \cite{Gope2019}.
%%%%%%%%%%%%%%%%%%%%%%%%%%%%
%\begin{figure}[H]
%    \centering
%    \includegraphics[width=2.5in]{GenericIoT.pdf}
%    \caption{Generic IoT Model \cite{patel2018internet}.}
%    \label{fig:1}
%\end{figure}
%%%%%%%%%%%%%%%%%%%%%%%%%%%
An Elliptic Curve Cryptography (ECC) attracts security researchers due to its lightweight operations, less computation requirement, and limited memory consumption. The ECC has proven it's computation efficiency and robust security against traditional public-key cryptography mechanisms like RSA. The ECC operations like point multiplication replace the conventional discrete logarithm mechanisms based on exponential computation. Recently authors in \cite{ABBASINEZHADMOOD201847} and \cite{Roy2018} proposed a session key agreement scheme using the ECC for a sensing environment.    

The sensing devices deployed on "ground" collect data from the environment and transmit those data to the nearby home agent (gateway). The neighboring home agents can be a micro-controller, micro-processor, mobile towers, routers, or any data receiver device which integrates data from the sensing devices and forwards those data to the users via other internet devices. The recent study shows that the latest home agents also work like fog devices or edge computing devices capable of performing the local data processing and converting those unorganized data into organized raw data. In traditional IoT network, deployed sensing devices create a local cluster and communicate pieces of information with the cluster heads (CHs) thorough the short-range protocols like Zigbee, Z-Wave, Beacon or Bluetooth Low Energy (BLE). In some of the IoT model, the sensing devices communicate with the cluster head through a wired medium. The CH connects the gateway devices (GW) with the sensing devices. The gateway devices are resource capable devices that can perform complex security operations and can forward the received data to the IoT application users through a long-range internet protocol like IP or 6LoWPAN. In IoT, users can access stored data as well as realtime live data. Thus, the gateway devices transmit data to the cloud server for storage and processing or to the user for realtime monitoring. Secure storage and processing of the data in the cloud lead toward intelligent decision making through machine learning.
%%%%%%%%%%%%%%%%%%%%%%%%%%%%%%%%%%%%%%%%%%%%%%%%%%%
%\begin{figure}[H]
%    \centering
%    \includegraphics[width=2.5in]{hierarchicalmodel.pdf}
%    \caption{IoT Data Transmission Model}
%    \label{fig:2}
%\end{figure}
%%%%%%%%%%%%%%%%%%%%%%%%%%%%%%%%%%%%%%%%%%%%%%%
In numerous recently proposed key agreement schemes, the application users register with the gateway devices for each sensing device \cite{ABBASINEZHADMOOD201847, Roy2018, Zhou2019}, but let us take the realtime scenario in which there are thousands of sensing devices deployed on the ground. Users like the company owner or the cluster manager want to receive the data from each sensing device in realtime. Then, they need to register for each sensing device individually, which is not a practical and feasible solution. During deployment of the realtime scenario, we found that in the recently proposed schemes, the gateway needs to create a separate smart card for each of the sensing devices for each of the application users who require $N_u * N_{sd}$ registrations and $N_u * N_{sd}$ time gateway initial computations. Here $N_u$ shows the number of users, and $N_{sd}$ indicates the number of sensing devices. Thus, in this paper, we reap a novel solution for the problems mentioned earlier, using a bizarre concept of the Level-dependent Authentication (LDA).       
\subsection{Related Work}
\label{Subsec:RelatedWork}
\noindent In IoT security, authentication is one of the significant operations which assures mutual trust between the sensing devices, intermediate network devices, and the end-user devices. The secure authentication and key agreement achieve features like \textit{an anonymity and unlinkability, lower setup complexity, fewer computations and communication cost, natural access control, limited energy consumption, service availability, data confidentiality, communication integrity, nontraceability, secure ownership transfer, and mutual authentication}. In this subsection, we discuss recently published authentication schemes designed using ECC \cite{miller1985use} for the generic IoT model and other IoT applications like smart grid and smart home. 

\subsubsection{Authentication in Smart Home:} Recently, Shuai et al. published an authentication scheme for the smart home using an ECC \cite{SHUAI2019SmartHomeEcc}. In this Paper, the Registration Authority (RA) is a trusted entity that performs an initialization step and generates secret credentials for the sensing device $SD_j$ and the gateway node $GW$. The scheme proposed in  \cite{SHUAI2019SmartHomeEcc} is a two-factor authentication scheme in which the user makes use of the password and smart card to perform the login and authentication. The other authentication scheme was recently proposed by Lyu et al. \cite{Lyu2019SmartHomeECC} for the intelligent home using ECC. 

Authors in \cite{Lyu2019SmartHomeECC} put forward an authentication scheme which provides security against the traceability and useful for the uncertain internet services and environment like "If This Than That (IFTTT)." In the same paper, authors give a formal security analysis using a practical scyther tool. In 2018, Chifor et al.  \cite{chifor2018securitysmarthomeecccoap} proposed a unique authentication scheme for the "Fast IDntity Online (FIDO) model. In the FIDO model, the user does not use any authentication factors like a password. Still, it uses ECC generated parameters as keys stored by the trusted party and biometric-based access for those keys. The other authentication protocol for the smart home using a password was proposed by Naoui et al. in \cite{Naoui2019smarthomeecc}. Authors in \cite{Naoui2019smarthomeecc} proposed a  lightweight and secure password-based authentication scheme called "LSP-SHAP" for the smart home monitoring and management. 

\subsubsection{Authentication in Smart Grid:} In 2016, Jo et al. proposed an authentication mechanism for the smart grid using ECC. Jo et al. \cite{Jo2016smartgridecc} proposed an authentication scheme between a smart meter (SM), Data Collection Unit (DCU), Advanced Metering Infrastructure (AMI) using ECC based key pair generation. In 2017, Vahedi et al. proposed an ECC based authentication scheme for the grid in \cite{vahedi2017securesmartgridecc}. Authors in \cite{vahedi2017securesmartgridecc} proposed an authentication mechanism between a smart meter (which collects an energy consumption from smart appliances), a gateway (which aggregates the data from all smart meter) and operation center (which works as a bill generating location) using a Trusted Third Party (TTP). 

In 2018, Mahmood et al. \cite{MAHMOOD2018557} proposed an authentication scheme for the smart grid in which other registered users and TTP authenticate with the registered user.  The authors in \cite{MAHMOOD2018557} provided a security analysis for the proposed authentication using a widely adopted and practical security analysis tool "ProVerif." In 2019, Kumar et al. \cite{Kumar2019SmartGridECC} proposed an ECC based authentication scheme for the smart grid. The network model used by \cite{Kumar2019SmartGridECC} consists of the authentication between the Energy Utility Center (EUC) and the Smart Grid device (SG) using the Trusted Authority (TA). The authors performed a security analysis for their proposed scheme using Automated Verification of Internet Security Protocol and Application (AVISPA) tool and random oracle based RoR Model. Recently, other authors also presented authentication schemes using ECC for the smart grid in \cite{ZHANG2019smartgridecc, khan2019secure}.

\subsubsection{Authentication in Smart Healthcare:} In 2018, Jia et al. \cite{Jia2018health} proposed a key agreement scheme for smart health-care using an ECC. The authors in \cite{Jia2018health} provides an authentication mechanism for the fog based cloud service dependent network model in which the critical agreement materializes between a user-fog node (FN)-cloud service provider (CSP). In 2018, Wu et al. \cite{wu2018health} proposed an authentication scheme for the health-care model where the patient with the mobile device communicates with the nearby home agent to transmit the body data to the doctors. 

The authentication model, followed in \cite{wu2018health}, provides a critical agreement between the mobile user-foreign agent-home agent. The AVISPA tool is used for security analysis. Recently in 2019, Ever at al. \cite{ever2019health} proposed an anonymous authentication scheme for the Wireless Medical Sensor Network (WMSN) where the WMSN user receives a live and stored sensor data through the gateway. The authors in  \cite{ever2019health} provides a formal security analysis of the proposed scheme using a random oracle based model. In 2019, Sureshkumar et al. \cite{suresh2019health} proposed an authentication scheme using an ECC for the WMSN in which a sensor-equipped patient with the smart device transmits data to the user (doctor) through the gateway device and also store the data in the cloud. The Authors in \cite{suresh2019health} implemented the proposed scheme using a Linear Feedback Shift Register (LFSR). In the IoT based smart health-care system, the privacy of the patient's identity and confidentiality of the health data is the critical security aspects.  

\subsubsection{Authentication in Generic IoT:} In 2018, Wazid et al. \cite{wazid2018} proposed a User Authenticated Key Management Protocol (UAKMP) for the smart home IoT network. Authors in \cite{wazid2018} followed the user-gateway-sensor based network model and the random oracle based ROR model for the formal security analysis. In 2019, Das et al. \cite{Das2019} proposed a lightweight access control and key agreement protocol for the IoT environment (LACKA-IoT) using ECC. Recently, in 2019, Gope et al. \cite{Gope2019} proposed a privacy-preserving authentication scheme for the IoT devices using a Physical unclonable Function (PUF). The PUF provides a lightweight hardware implementation of the random number generator. In \cite{Das2019}, the ROR based security model was followed, and the security simulation is produced using the AVISPA tool. They simulated the proposed protocol using a widely used simulator Network Simulator 2 (NS2). 

A publish-subscribe based MQTT protocol is widely accepted for IoT based applications. The Lohachab et al. \cite{Lohachab2019}, in 2019, proposed an ECC based authentication and access control scheme for the MQTT based communications. The Machine to Machine communication through the MQTT protocol plays a significant role in automated service developments. The authors in \cite{Esfahani2019} proposed an authentication scheme between the sensing device and routing device (device-device) using lightweight operations like hash function and XOR operation. The Internet of Drones (IoD) is a network of uncrewed areal vehicles called Drones. In IoD based interface, the secure live streaming and reliable access control of the drone devices are essential security aspects. The authors in \cite{Wazid2019drone} set forth a crucial lightweight agreement scheme for the drone deployment in which the ground user securely communicates with the Drone Data Transmitter (DDT) through the server as a trusted entity. 

Due to the widespread growth of the IoT based devices and their deployments, the attack space is also expanded as well as new attack vectors are also created for the attackers. The layered wise security analysis of the IoT is briefly discussed in \cite{Mosenia2017}. The physical attack, side-channel, and DoS attack are major attacks through which IoT nodes pass through. The side-channel attack faced by IoT devices, user devices, or network devices, which can leak certain critical information such as communication time, communication frequency, communication direction, communication modulation to the attacker. The significant attacks through which the IoT authentication schemes pass through are \textit{replay attack, Sybil attack, flooding attack, DoS attack, stolen smart card attack, forward secrecy, password guessing attack, eavesdropping, forged user attack,  masquerading server attack, man-in-the-middle attack, stolen verifier attack, session-specific secret loss attack, device compromised and device impersonation attack, insider attack and so on.} 

\subsection{Research Contribution}
\label{Subsec:Contribution}
\noindent There are multitudinous research contributions from this paper. 
\begin{itemize}
    \item In this paper, we provide a solution for the problem of "multiple registrations by the single user for the different sensing devices" using a Level-Dependent Authentication (LDA). The LDA algorithm is straightforward but highly efficient to use in the massive IoT industrial deployment. The LDA protocol significantly reduces the access control complexity for large industries. To the best of our knowledge, the LDA is a novel and unique concept that is proposed for the first time in this paper.
    \item The proposed LDA scheme uses only hash-functions ECC based computations and ECC encryption/decryption function, which make the proposed scheme a  lightweight compared to the other existing schemes. 
    \item The security analysis of the proposed LDA scheme is performed in two-fold. One fold is through the informal way; we prove that the proposed LDA scheme is secured against various traditional and non-traditional IoT attacks over the Dolev-Yao channel. In the second fold, we prove the security of the proposed scheme using a widely accepted and recognized random oracle based security model "Real or Random (RoR) model." Along with the RoR, we also provide security simulation for the proposed scheme using an AVISPA tool. The security standardization authorities widely recognize the AVISPA tool like IETF and others for their security simulations.
    \item The comparative analysis of the proposed scheme with the other existing systems proves that the proposed LDA scheme is more reliable, efficient, and convenient in-terms of computation cost, communication cost, energy consumption, and deployment in a massive industry. The implementation validity of the proposed scheme is verified through the deployment of NodeMCU as sensing devices, raspberry-pis as gateway devices and laptops as a user device. The proposed LDA-2IoT scheme is performed using the MQTT as an application layer protocol and 6LoWPAN as a network layer protocol. 
\end{itemize}
\subsection{Paper Organization}
\label{Subsec:Organization}
\noindent The rest of the paper is organized as follows: In section \ref{Sec: Preliminaries}, we provide a brief introduction about the network model followed for designing of the proposed scheme. We introduce readers with the basic preliminaries like the one-way hash function, ECC, level-dependent authentication, and threat model. Section \ref{Sec:ProposedScheme} presents a proposed LDA-2IoT scheme using an ECC. In section \ref{Sec:SecurityAnalysis}, we provide a detailed security analysis for the proposed scheme using the formal model and the informal method. The performance analysis for the proposed scheme and its comparison with the other existing schemes is discussed in section \ref{Sec:performanceanalysis}. In section \ref{Sec:Implementation}, we provide the implementation methodology and the output for a proposed scheme. Lastly, we conclude the article with the future directions in section \ref{Sec:Conclusion}.   

\section{PRELIMINARIES AND THREAT MODEL}
\noindent In this Section, we put forward the essential preliminaries and threat model used for the designing of the proposed protocol.
\label{Sec: Preliminaries}
\subsection{System Model}
\label{Subsec:NetworkModel}
\noindent The generic IoT system is a network of heterogeneous tiny resource constraint devices. In generic IoT, the end-user wants a data sensed by the resource constraint sensing devices. Therefore, in general, any IoT system provides two types of data services. The first type of facility where the user wants quick realtime live data. Examples of this type of application are smart disaster management, intelligent home/industry monitoring, smart energy monitoring, and so on. The second type of service where the user does not want realtime live data, but he/she retrieves stored data for analysis and smart decision making. Examples of this type of application include smart automated decision-making, intelligent recommender systems, intelligent learning-based security mechanisms, and so on. Thus, we tried to design and implement a proposed authentication scheme in such a way that it can be used for both the scenario. We discuss the system model in two ways. One is a network model where we highlight the topology on network and second as a communication model where we present the discussion about network layer and application layer protocols. 
\subsubsection{Network Model}
\noindent As shown in Fig. \ref{fig:3}, the network model involves four basic entities, the sensing device ($SD$), the cluster head ($CH$), the user devices ($U$) and the gateway node ($GWN$).  
\begin{itemize}
    \item \textit{User Device ($U$):} These devices have the end-user application installed, which provides a dashboard to the user for monitoring and controlling of the deployed system. The user device must be able to capture live data as well as stored data. The user device is built-in with limited resources and can still perform the basic cryptography operations. The smart card reader is a software program that is installed with the user devices that also communicates with the gateway device during authentication.
    \item \textit{Gateway Node ($GWN$):} The gateway node is a trusted party for the proposed IoT topology. The gateway node is a resource capable device and can perform complex cryptographic operations. The gateway device receives data from the sensing devices and forwards those data to the user devices after completing the user verification and data validation. Thus, the gateway device works as an aggregator, as well as a forwarder of the data. In the fog computing/edge computing concept, the gateway devices perform as a critical service provider. In the proposed scheme, we focus on the uni gateway model. In the subsequent future work, we plan to implement a multi-gateway model with a more realistic approach to complete IoT deployment. 
\begin{figure}[H]
    \centering
    \includegraphics[width=2.5in]{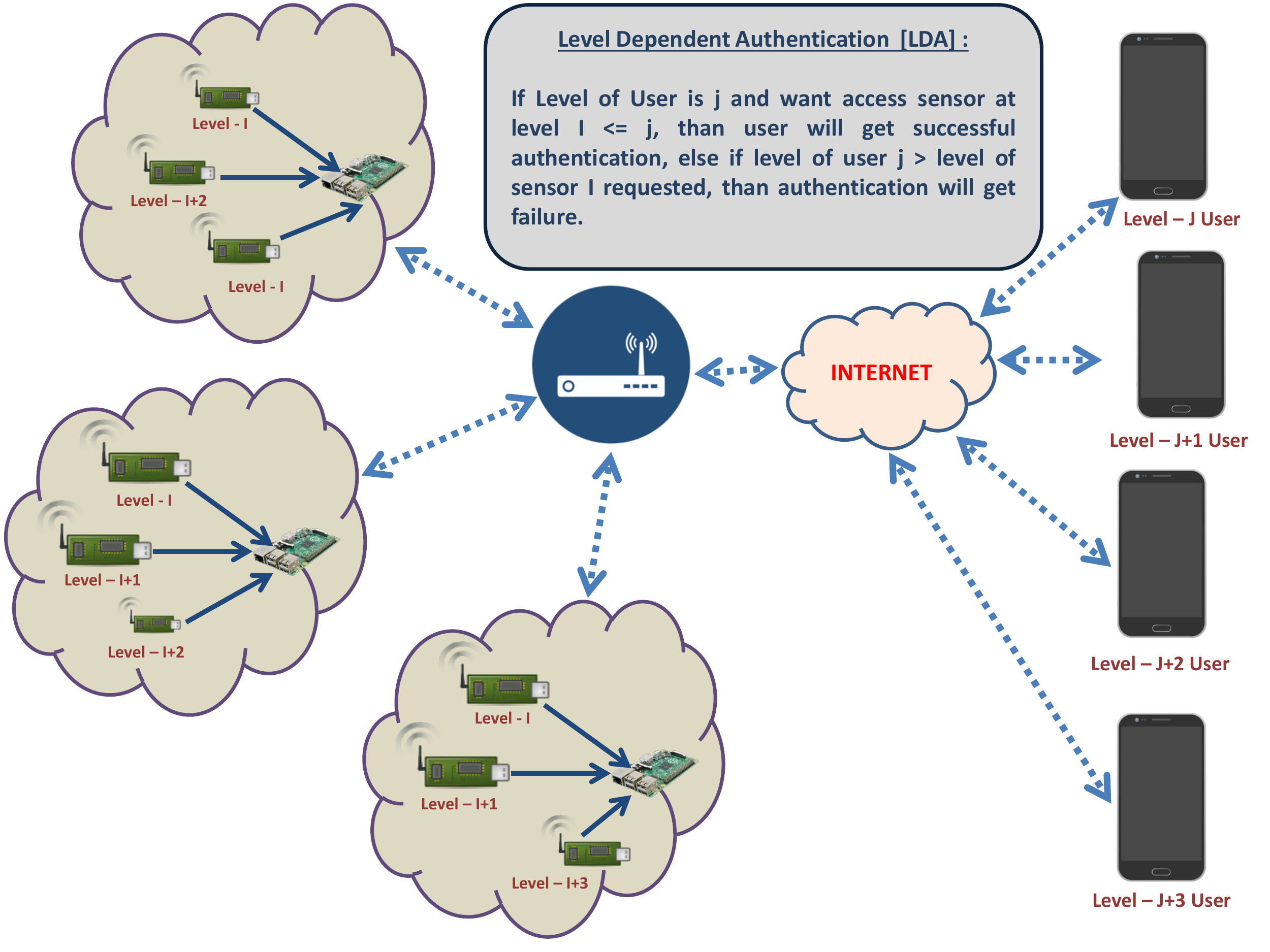}
    \caption{Network Model and Level Dependent Authentication}
    \label{fig:3}
\end{figure}
    \item \textit{Sensing Device ($SD$):} The sensing devices are tiny, and most resource constraint devices in the IoT hierarchy. These devices are highly confined in terms of memory, transmission bandwidth, communication range, computational capability, and power capability. Thus, the security operations performed by these devices must be lightweight and efficient enough. 
    \item \textit{Cluster Heads ($CH$):} The cluster head is a device that works as a networking hub or a switch that receives the data from multiple sensing devices deployed in-the network and forwards those data to the nearby in range gateway device. 
\end{itemize}
In the proposed network model, we assume that the gateway device is the most trusted, highly computationally capable, and physically secured machine \cite{Kumar2015home} \cite{Wazid2018secure}. Thus, the essential responsibilities like initialization, key generation, level verification are taken care of by the gateway devices. 
\subsubsection{Communication Model} In any IoT network, communication is performed in two ways: long-range communication and short-range communication. The short-range communication between the sensing device-cluster head or sensing device-gateway node uses protocols like Z-Wave, Zigbee, Beacons, Bluetooth Low Energy (BLE). The long-range communication between two gateway devices, CH device-gateway device, or gateway device-user device, uses networking protocols like IPv4, IPv6, 6LoWPAN, etc. The primary application layer protocols used by the IoT network are Message Queuing Telemetry Transport (MQTT), Light HTTP, Constrained Application Protocol (COAP), and Extensible Messaging and Presence Protocol (XMPP). 
\subsection{Level-Dependent Authentication}
\label{Subsec:LDA}
\noindent To achieve successful access control in the IoT scenario is also one of the principal challenges. Essential factors that affect the liberty of access control is the availability of the tonnes of sensing devices and heterogeneity of their technical capabilities. Thus, for every user, if we maintain the access control list or perform the registration, then it will be an unvaried task. As per the current literature, the user registers for individual sensors. If the user is eligible for hundreds of sensing devices, then he/she needs to maintain hundreds of smart cards. This is a significant challenge to reduce the space requirement with less complexity and wipe out the user's multiple registrations.  

Rather than the traditional approach, we highlight a novel concept of the Level-Dependent Authentication (LDA) to tackle the above said challenges using a less computation cost, low energy consumption, fewer operations, and little memory requirement. The working of the LDA concept is highlighted in Fig. \ref{fig:3}. The following algorithm provides a working mechanism of the LDA concept. 

\begin{algorithm}
\SetAlgoLined
\KwResult{Access of Sensing Device to User}
 \textit{User-level} = i\;
 \textit{Sensor-level} = j\;
 \While{Gateway received request from user}{
  \eIf{j $\leq$ i}{
   \textit{Access-Allowed}\;
   }{
   \textit{Access-Not-Allowed}\;
  }
 }
 \caption{Level Dependent Authentication}
\end{algorithm}

Therefore, the major advantages of using LDA can be listed as follows:
\begin{itemize}
    \item Reduction in the access control complexity.
    \item Number of registration phases and initialize phases will be reduced to the number of users rather than the number of sensing devices.
    \item Smooth replacement of the user device, sensing device, and gateway device compare to the existing traditional approach.
    \item Reduction of the computation cost, energy consumption, and memory utilization at user devices, gateway devices, and sensing devices.
\end{itemize}

The major challenges and future research directions related to LDA concept can be listed as follows:
\begin{itemize}
    \item Little increase in the computation cost at the gateway device.
    \item In the proposed scheme, the decision for users and sensors' level will be taken by gateway. The proposed scheme can get further extension where separate authority like registration authority can be created for the purpose of initial parameter computation as well as level decision. 
\end{itemize}
%%%%%%%%%%%%%%%%%%%%%%%%%%%%%%%%%%%%%%%%%%%%%%%%%%%%%%%%%%%%%%%%%%%%%%%%%%%%%%%%%%%%%%%%%%%%%%%%%%%%%%%%%
%%%%%%%%%%%%%%%%%%%%%%%%%%%%%%%%%%%%%%%%%%%%%%%%%%%%%%%%%%%%%%%%%%%%%%%%%%%%%%%%%%%%%%%%%%%%%%%%%%%%%%%%%
We validated the LDA algorithm for the MQTT based communication environment with user devices, gateway devices, and the sensing devices. We created a secure channel through the MQTT over TLS. We deployed numerous sensing devices on the university campus at different locations like computer labs, faculty cubicles, canteens, and admin offices (total fifty sensing devices). We considered the user's level, such as level 1 for the director, level 2 for the deans, level 3 for the faculties, level 4 for the admin staff, and level 5 for the clerical staff. As a gateway device, we used Raspberry pi 3 Model B with 1 GB RAM. The proposed scheme is verified, and we firmly realized that the proposed LDA optimizes the authentication process by storage, memory, and computation significantly.

\subsection{One-way Hash Function}
\label{Subsec:HashFunction}
\noindent The one-way hash function is a very essential and useful cryptographic operation that takes an arbitrary length of input and produces a fixed-length output. Thus, we can define the hash function H:$\{0,1\}^*$ $\Rightarrow$ $\{0,1\}^n$. Any hash function is deterministic, and the output of the hash function is also called a message digest or hash output. So for any binary message m $\in$ $\{0,1\}^*$, the message digest n $\in$ $\{0,1\}^n$ can be computed as n = H(m). 

As defined in \cite{wazid2017secure-smarthome, wazid2018}, for any adversary $\mathcal{A}$, \textit{$Adv_A^{Hash}$(et): Prob[Rand(a,b), $\mathcal{A}$: $p \neq q$ and H(p) = H(q)]} denotes advantage of $\mathcal{A}$ in finding a hash collision where, Prob[$\chi$] highlights a probability for the event $\chi$ and Rand(a,b) denotes that a and b are randomly generated by an adversary $\mathcal{A}$ and \textit{et} is the execution time. The adversary $\mathcal{A}$ is a probabilistic adversary whose advantage is decided by the random choices done by the adversary $\mathcal{A}$ within an execution time \textit{et}. Thus, if A's running time is at most \textit{et} then ($\rho$, t) shows $\mathcal{A}$ attacking the collision resistance of H(.). Thus, within maximum run-time \textit{et}, $Adv_A^{Hash}(et) \leq \rho$.

\subsection{Elliptic Curve Cryptography}
\label{Subsec:ECC}
\noindent An Elliptic Curve Cryptography (ECC) is an example of the lightweight public-key cryptography. The comparative analysis of the ECC and other public-key cryptography protocols like RSA is briefly discussed in \cite{patel2018internet} by Patel et al. The 160 bit of the ECC key provides a security equivalent to 1024 bit of RSA key. An Elliptic curve is a cubical curve with the non-repeatable roots defined over a finite field \textit{F(p)} where \textit{p} is a prime number greater than three. A curve is represented as a \textit{(x,y) $\in$ F(p)*F(p)} for the equation,

\begin{equation}
    y^2 = x^3+ax+b mod p
\end{equation}

where \textit{$4a^3+27b^2 \neq$ 0 mod p}. If \textit{$G_p$($G_x$,$G_y$)} is a base point or generator point than any point \textit{($x_p$,$y_p$)} generated by the \textit{$G_p$} will be on curve along with the point of infinity "$\mathcal{O}$".

%%%%%%%%%%%%%%%%%%%%%%%%%%%%%%%%%%%%%%%%%%%%%%%%%%%%%%%%%%%
%\begin{figure}[H]
%    \centering
%    \includegraphics[width=2.5in]{ellipticoperation.pdf}
%    \caption{Basic Operations on Elliptic Curve}
%    \label{eccoperation}
%\end{figure}
%%%%%%%%%%%%%%%%%%%%%%%%%%%%%%%%%%%%%%%%%%%%%%%%%%%%%%%%%%%%
For any two curve points, \textit{P($X_p$,$Y_p$)} and \textit{Q($X_q$,$Y_q$)} the sum \textit{R = P + Q} which is \textit{R($X_r$,$Y_r$)} where \textit{$P \neq -Q$}. For any \textit{$\Lambda$ = ($\frac{Y_q-Y_p}{X_q-X_p}$) mod p} if \textit{P $\neq$ Q} and \textit{$\Lambda$ = ($\frac{3X_p^2 + a}{2Y_p}$) mod p} if \textit{P = Q}.  \textit{$X_r$ = ($\Lambda^2$ - $X_p$ - $X_q$) mod p} and \textit{$Y_r$ = ($\Lambda(X_p - X_r)-Y_p$) mod p}. 
\noindent \textit{\textbf{ECC Encryption:}} The ECC encryption invokes an encoding for the message \textit{$m$} in to the curve point \textit{$P_m$}. For any random private key \textit{$K_x$} generated by the user \textit{$U_x$},  the relative public key \textit{$KP_x$ =  $K_x$ * $G_p$}, where \textit{$G_p$} is any group point on the elliptic curve. To encrypt the \textit{$P_m$}, user \textit{$U_x$} selects the random number \textit{k} and computes \textit{$C_m$ = {(k*$G_p$, $P_m$ + k * $KP_y$)}} where \textit{$KP_y$} is a public key of the receiver \textit{$U_y$}. User \textit{$U_x$} sends \textit{$C_m$} to \textit{$U_y$} over a public channel.   

\noindent \textit{\textbf{ECC Decryption:}} The ECC decryption invokes a computation for the \textit{$P_m$ =  $P_m$ + k * ($K_y$ * $G_p$) - ($K_y$*(k*$G_p$))} where \textit{$K_y$} is Y's private key. The advantage of adversary \textit{$\mathcal{A}$} in computing \textit{k} from the \textit{k*$G_p$} can be defined as,  \textit{$Adv_A^{Dec}$(et) = Pr[Rand(k, $G_p$), $\mathcal{A}$: $\chi_x$]} where \textit{$Adv_A^{Dec}$(et) $\leq$ $\rho$}, for any \textit{$\rho$ $\textgreater$ 0} and randomly generated pair \textit{(k, $G_p$)} with execution time \textit{et} in such a way that \textit{$\chi_x$ = k*$G_p$}.

\subsection{Threat Model}
\label{Subsec:ThreatModel}
\noindent The threat model used in this paper is contemplated from the homogeneous model discussed in \cite{Dolev1981}. An adversary $\mathcal{A}$ is an eavesdropper who controls the complete public communication channel. In the IoT based network model, it is possible to define an adversary $\mathcal{A}$ with the robust capabilities for improvement in the designing of the reliable protocol and also to perform the better security validation for the proposed authentication concept. We follow the following adversarial model in this paper:  
\begin{enumerate}[$G_1$.]
    \item An Adversary $\mathcal{A}$ can compute valid pair of the $identity * password$ offline in polynomial time using dictionary \cite{wazid2018}, \cite{Gope2019}.
    \item An Adversary $\mathcal{A}$ can extract the data from the user's smart card after receiving smart card in either ways \cite{Gope2019}, \cite{SHUAI2019SmartHomeEcc}. 
    \item An Adversary $\mathcal{A}$ have full access on the communication channel between a User - Gateway, Sensor node - Gateway, and User - Sensor node \cite{wazid2018}, \cite{Gope2019}, \cite{SHUAI2019SmartHomeEcc}.
    \item An Adversary $\mathcal{A}$ can get the previously computed session key between the user and sensor. $\mathcal{A}$ can use this key to compute the next session key \cite{Kumar2015home}. 
    \item An Adversary $\mathcal{A}$ can have the level information of the user device or the sensing device at a time but can't have the level of both at a time \cite{vahedi2017securesmartgridecc}, \cite{SHUAI2019SmartHomeEcc}. 
    \item An Adversary $\mathcal{A}$ can have the secrets of a gateway node during the system failure situations. $\mathcal{A}$ can use this old secrets to break the newly established system after failure.
    \item An Adversary $A$ can perform the physical attacks on sensor nodes and can retrieve the information stored into it \cite{Kumar2019SmartGridECC}, \cite{Kumar2015home}.
    \item An Adversary $\mathcal{A}$ can generate bot nodes and can send the simultaneous ping messages to the sensor node with the aim to perform DoS attacks \cite{Kumar2015home}. 
\end{enumerate}
\section{LEVEL-DEPENDENT AUTHENTICATION FOR GENERIC IoT (LDA-2IoT)}
\label{Sec:ProposedScheme}
\noindent In this section, we put forward the proposed Level-Dependent Authentication Scheme for Generic IoT (LDA-2IoT). We offer an LDA-2IoT between the user device and the sensing device through the intermediary gateway node. As earlier said, we consider the gateway device as a trusted and secure node. We assume that the universal clock for all the devices in the system is synchronized. The proposed vital agreement scheme consists of three phases: system initialization phase, user registration phase, and login and key-agreement phase. We consider the gateway as a master device, and the key-agreement is also going to carry through the gateway device. The necessary notations used for designing of the proposed scheme are highlighted in Table \ref{tab:notation}.  
\begin{table}[H]
    \centering
    \caption{Symbols and Notations}
    \begin{tabular}{p{8em}p{16em}} \hline
    Symbols & Description \\ \hline
    $R_x$ & Random Number \\
    $T_x$ & Time-stamp \\ 
    $l_i$ & User Level \\ 
    $l_j$ & Sensor Level \\ 
    $U_i$ & User Device \\
    $S_j$ & Sensor Device \\
    $GW$ & Gateway Node \\
    $SID_j$ & Sensor Identity \\
    $UID_i$ & User Identity \\
    $GWID$ & Gateway Identity \\
    $G_p$ & Elliptic Curve Generator \\ 
    $\Delta T$ & Time-stamp Threshold \\ 
    $K_s$ & Gateway Node Master Secret \\ 
    H(.) & One-way Hash Function \\
    Enc(.)/Dec(.) & ECC Encryption/Decryption \\
    $\bigoplus$, $||$ & XOR and Concatenation Respectively \\ \hline 
    \end{tabular}
    \label{tab:notation}
\end{table}
\subsection{System Initialize Phase}
\label{Subsec:IntializePhase}
\noindent  In this subsection, we discuss the system's initialization phase. All steps in the initialization phase of the system are carried out by the gateway node in an offline manner. Thus, message generation and message communication in this phase occur in a secure environment. The gateway device computes parameters for the user devices and sensing devices. The gateway device decides level for the user device based on the position of the user in an organizational hierarchy and the level of sensing device based on its location of deployment in the environment. It is necessary to observe that none of the devices store their levels in any format.
\subsubsection{Gateway Initialize Phase}
The gateway initialize phase occurs as follow,
\begin{itemize}
    \item Generates random private key $RGWN_k$ from the range of 1 to n where n is the large prime order of the elliptic curve. 
    \item Generates a gateway random master key $K_s$. 
    \item Computes gateway node public key as a $PUB_{GW_k}$ = $RGWN_k$ * P, where P is the curve point.
\end{itemize}
\subsubsection{User Device Initialize Phase}
The user device initialize phase occurs as follow,
\begin{itemize}
    \item Generates a random private key for each $i^th$ user as $RU_i$ from the range of 1 to n where n is a large prime order of the curve, and i ranges from 1 to the number of users in the IoT network. 
    \item Computes public key for the user $U_i$ as a $PUB_{U_i}$ = $RU_i$ * P, where P is curve point.
    \item Generates random identity for each user $U_i$ as $UID_i$.
    \item Computes $X_1$ = H($RU_i || UID_i || K_s$).
    \item Computes $X_2$ = H($UID_i|| PUB_{U_i}||K_s$).
    \item Stores $X_1$, $X_2$, $RU_i$ in the secret memory of the user $U_i$. 
\end{itemize}
\subsubsection{Sensor Device Initialize Phase}
The sensor device initialize phase occurs as follow,
\begin{itemize}
    \item Generates a random number as a private key for each sensor node $S_j$ called as a $RSN_j$. 
    \item Computes public key for the sensing device $S_j$ as a $PUB_{S_j}$ = $RSN_j$ * P, where P is the curve point.
    \item Generates random identity for each sensor node $S_j$ as a $SID_j$.
    \item Computes $Y_1$ = H($RSN_j || SID_j || K_s$).
    \item Computes $Y_2$ = H($SID_j || PUB_{S_j} || K_s$).
    \item Computes $D_j$ = H($l_j || K_s || SID_j)$ where \textbf{$l_j$} is the level of jth sensor based on its deployment in network.
    \item Stores $Y_1$, $Y_2$, $RSN_j$, $D_j$ in the secret memory of the sensing device $S_j$. 
\end{itemize}
Gateway node fly parameters $PUB_{GW_k}$, $PUB_{U_i}$, $PUB_{S_j}$ as a public parameters. We like to point out that during the implementation process of the proposed LDA-2IoT, we stored all these parameters in all the devices as a publicizing process. 
\subsection{User Registration Phase}
\label{Subsec:UserRegPhase}
\noindent In this section, we discuss the user registration process carried out in a secured manner between user device and the gateway device. The user registration phase follows following steps:
\begin{enumerate}
    \item \textbf{\textit{$U_i$ $\xrightarrow{Request}$ $GW$:}} The user $U_i$ selects the password $UPW_i$, generates the random numbers $R_a$, $R_b$, computes the $TPW_i$ = H($UPW_i||R_a$) $\bigoplus$ $R_b$ and sends \textit{Request} = \{$UID_i$\} to the gateway $GW$.
    \item \textbf{\textit{$GW$ $\xrightarrow{Smart Card}$ $U_i$:}} The gateway computes, $Reg_i$ = H($UID_i||K_s$), computes $B_i$ = H($l_i || K_s || UID_i)$ where \textbf{$l_i$} is the level of ith user based on its role in the organization and $K_s$ is the gateway master secret. Generate smart card \textit{SC} = \{$Reg_i$, $B_i$, H(.), $E_p(a,b)$\} and sends to the user $U_i$.
    \item The user computes $L_1$ = $H(UID_i)$ $\bigoplus$ $R_a$, $TPW_i'$ = $TPW_i$ $\bigoplus$ $R_b$, $L_2$ = H($UID_i||TPW_i'$), $Reg_i*$= $Reg_i$ $\bigoplus$ $R_b$, replaces $Reg_i$ by $Reg_i*$ in SC and creates final SC = \{$Reg_i*$, $L_1$, $L_2$ ,$B_i$, H(.), $E_p(a,b)$\} 
\end{enumerate}
\subsection{Login and Session Key Agreement Phase}
\label{Subsec:Loginkeyphase}
\noindent In this subsection, we discuss two phases, the login phase and the session key agreement phase in which the user device $U_i$ wants to access the data from the sensing device $S_j$, and for that, it tries to establish a session key with the $S_j$. In the login phase, the user provides $UID_i$, $UPW_i$ and SC to the Smart Card Reader (SCR), the SCR verifies all the parameter and computes new parameters for the key agreement phase. All the steps of the session key agreement phase perform through the public channel. In this phase, $U_i$ sends a request to the $GW$. The $GW$ verifies level and other parameters of the $U_i$ and prove it's access capabilities. Later on, through the $GW$ device, $U_i$ and $S_j$ generates a mutually authenticated session key $SK$. The login phase and the session key agreement consist of the following steps:
\begin{enumerate}
    \item \textbf{\textit{$U_i$ $\xrightarrow{Request}$ $SCR$:}} The user provides $UID_i$ and $UPW_i$ and SC to the SCR. The SCR computes $R_a*$ = $L_1$ $\bigoplus$ H($UID_i*$), $TPW_i*$ = H($UPW_i||R_a*$), $L_2*$ = H($UDI_i||TPW_i*$) and verifies $L_2*$ = $L_2$.  If verification gets success, SCR allows $U_i$ for the further key agreement else abort the procedure. 
    \item \textbf{\textit{$U_i$ $\xrightarrow{Message 1}$ $GW$:}} The user device $U_i$ gets current time-stamp $T_1$, random $r_t$ and computes $M_1$ = $Enc_{PUB_{GW_k}}$ ($Temp_0$, $Pub_{U_i}$,$PUB_{S_j}$,$r_t$,$B_i$), $Temp_0$ = H($X_2||T_1||r_t$). The $U_i$ sends \textit{Message 1} =  \{$M_1$,$Temp_0$ $T_1$\} to $GW$. 
    \item \textbf{\textit{$GW$ $\xrightarrow{Message 2}$ $S_j$:}} The gateway device $GW$ gets current time-stamp $T_1*$ and verifies $\Delta$T $\leq$ $T_1*$ - $T_1$. Gets \{H($X_2||r_t||T_1$), $r_t$, $Pub_{U_i}$,$PUB_{S_j}$,$B_i$\} = $Dec_{RGWN_k}$($M_1$), extracts valid $UID_i*$ for $Pub_{U_i}$ from it's secret memory  and verifies H($H(PUB_{U_i}||UID_i*||K_s)||T_1||r_t$)$\stackrel{?}{=}$ $Temp_0$. If yes, Move on. Get current timestamp $T_2$, $SID_j$ using $PUB_{S_j}$ and computes $Temp_1$ = H($PUB_{S_j}||H(SID_j||K_s||PUB_{S_j})||PUB_{GW_k}||T_2$). The $GW$ sends \textit{Message 2} = \{$Temp_1$,$T_2$\} to $S_j$.
    \item \textbf{\textit{$S_j$ $\xrightarrow{Message 3}$ $GW$:}} The sensing device $S_j$ gets current timestamp $T_2*$ and verifies $\Delta$T $\leq$ $T_2*$ - $T_2$. Verifies H($PUB_{S_j}||Y_2||PUB_{GW_k}||T_2$)$\stackrel{?}{=}$ $Temp_1$. After successful verification, $S_j$ gets current timestamp $T_3$ and computes $M_2$ = H($Y_2||T_3||PUB_{S_j}$), $M_3$ = $Enc_{PUB_{GW_k}}$($M_2$, $D_j$). The $S_j$ sends \textit{Message 3} = \{$M_3$,$T_3$\} to $GW$.   
    \item \textbf{\textit{$GW$ $\xrightarrow{Message 4}$ $S_j$:}} The gateway device $GW$ gets current timestamp $T_3*$ and verifies $\Delta$T $\leq$ $T_3*$ - $T_3$. Get \{$M_2$, $D_j$\} = $Dec_{RGWN_k}$($M_3$) and verifies H($PUB_{S_j}||T_3||H(SID_j||K_s||PUB_{S_j})$)$\stackrel{?}{=}$ $M_2$. Gets $l_i$ and $l_j$ from $B_i$ and $D_j$ respectively by computing: $B_i*$=H($l_i||K_s||H(UID_i)$) till  $B_i*$ $\stackrel{?}{=}$ $B_i$ satisfies for valid $l_i$ and $D_j*$=H($l_j||K_s||H(SID_j)$) till  $D_j*$ $\stackrel{?}{=}$ $D_j$ satisfies for valid $l_j$. Now the $GW$ verifies if $l_i \leq l_j$, then continues else transmits 0 signal to $U_i$, $S_j$ and abort the connection. The $GW$ generates random number $r_1$ and gets current timestamp $T_4$. The $GW$ computes $M_4$ = H($PUB_{U_i}||T_4||PUB_{GW_k}||r_1||H(PUB_{S_j}||SID_j||K_s)$), $M_5$ = $Enc_{PUB_{S_j}}$($M_4$, $SID_j$,$r_1$). The $GW$ sends \textit{Message 4} = \{$M_5$,$PUB_{U_i}$,$T_4$\} to $S_j$.    
    \item \textbf{\textit{$S_j$ $\xrightarrow{Message 5}$ $GW$:}} The sensing device $S_j$ gets current timestamp $T_4*$ and verifies $\Delta$T $\leq$ $T_4*$ - $T_4$. Gets \{$M_4$,$SID_j$,$r_1$\} = $Dec_{RSN_j}$($M_5$) and verifies H($PUB_{U_i}||T_4||PUB_{GW_k}||r_1||H(Y_2)$) $\stackrel{?}{=}$ $M_4$. if yes, move on. The $S_j$ generates random number $r_2$, gets current timestamp $T_5$ and computes $M_6$ = H($PUB_{S_j}||PUB_{GW_k}||PUB_{U_i}||r_2||T_5$), $M_7$ = H($r_1||M_6||SID_j||PUB_{S_j}||T_5$). The $S_j$ sends \textit{Message 5} = \{$M_6$,$M_7$,$T_5$\} to $GW$.
    \item \textbf{\textit{$GW$ $\xrightarrow{Message 6}$ $U_i$:}} The $GW$ gets current time-stamp $T_5*$ and verifies $\Delta$T $\leq$ $T_5*$ - $T_5$. The $GW$ verifies H($r_1||M_6||SID_j||PUB_{S_j}||T_5$) $\stackrel{?}{=}$ $M_7$. If yes, Move on. The $GW$ generates random number $r_3$ and gets current timestamp $T_6$. The $GW$ computes $M_8$ = H($r_3||M_6||M_7||T_6$), $M_9$ = H($SID_j||PUB_{S_j}||K_s$), $M_{10}$ = H($PUB_{U_i}||UID_i||K_s$), $M_{11}$ = H($PUB_{U_i}||UID_i||PUB_{GW_k}||T_6$), $M_{12}$ = H($PUB_{GW_k}||SID_j||PUB_{S_j}||T_6$), $M_{13}$ = $Enc_{PUB_{U_i}}$ ($M_8$,$M_9$,$M_{11},r_1,r_2$), $M_{14}$ = $Enc_{PUB_{S_j}}$ ($M_8$,$M_{10}$,$M_{12},r_t$). The $GW$ sends \textit{Message 6} = \{$M_{13}$,$T_6$\} to $U_i$.  
    \item \textbf{\textit{$GW$ $\xrightarrow{Message 7}$ $S_j$:}} The $GW$ sends \textit{Message 7} = \{$M_{14}$,$T_6$\} to $S_j$.
    \item The $U_i$ gets current timestamp $T_6*$ and verifies $\Delta$T $\leq$ $T_6*$ - $T_6$. Gets \{$M_8$,$M_9$,$M_{11},r_1,r_2$\} = $Dec_{RU_i}$($M_{13}$), verifies H($PUB_{U_i}||UID_i||PUB_{GW_k}||T_6$) $\stackrel{?}{=}$ $M_{11}$ and computes session key \textbf{SK} = H($M_8||M_9||T_6||X_2||r_1||r_2||r_t$).
    \item The $S_j$ gets current timestamp $T_6*$ and verifies $\Delta$T 
    $\leq$ $T_6*$ - $T_6$. Gets \{$M_8$,$M_{10}$,$M_{12},r_t$\} = $Dec_{RU_i}$($M_{14}$), verifies H($PUB_{S_j}||SID_j||PUB_{GW_k}||T_6$) $\stackrel{?}{=}$ $M_{12}$ and computes session key \textbf{SK} = H($M_8||M_{10}||T_6||Y_2||r_1||r_2||r_t$).
\end{enumerate}

%%%%%%%%%%%%%%%%%%%%%%%%%%%%%%%%%%%%%%%%%%%%%%%%%%%%%%%%%%%%%%%%%%%%%%%%%%%%%%%%%%%%%%%%%%%%%%%%%%%%%%%%%%%%%%%%%%%%%
\section{SECURITY ANALYSIS}
\label{Sec:SecurityAnalysis}
\noindent In this section, we provide the security analysis for the proposed LDA-2IoT. Security comparison of the proposed scheme with existing schemes shown in Table \ref{tab:securitycomparison}.
\subsection{Informal Security Analysis using Dolev-Yao Channel}
\label{Subsec:InformalAnalysis}
\noindent The Dolev-Yao channel \cite{Dolev1981} is a communication model based on \textit{snd} and \textit{rcv} operations. In this subsection, we set forth the informal security analysis for the proposed protocol based on a Dolev-Yao channel. The polynomial time adversary $\mathcal{A}$ can access and control the Dolev-Yao channel. In the proposed scheme, we consider that the initialize phase implemented over the secure channel, and the gateway device is a trusted secure device. In this subsection, we discuss how the proposed system provides security against the most well-known attacks. 
\subsubsection{Anonymity and Tracebility}
The anonymity for the security algorithm assures that an identity of the user is secured against the adversary's knowledge. In the initialize phase of the proposed scheme, the trusted $GW$ generates an identity of the $i^th$ user as $UID_i$ and $j^th$ sensing device as $SID_j$. Later on $GW$ computes $X_1$ = H($RU_i||UID_i||K_s$) and $X_2$ = H($UID_i|| PUB_{U_i}||K_s$) for each $U_i$. During the login and key-exchange phase, user communicates message $M_1$ = $Enc_{PUB_{GW_k}}$(H($X_2||T_1$), $r_t$, $Pub_{U_i}$, $PUB_{S_j}$, $B_i$) which is secured through the public-key of gateway. Now, let us assume that an adversary $\mathcal{A}$ intercepts other messages $Temp_1$ = H($PUB_{S_j}||SID_j||PUB_{GW_k}||T_2$), $M_3$ = $Enc_{PUB_{GW_k}}$($M_2$, $D_j$),  $M_5$ = $Enc_{PUB_{S_j}}$($M_4$,$r_1$), $M_6$ = H($PUB_{S_j}||PUB_{GW_k}||PUB_{U_i}||r_2||T_5$), $M_7$ = H($r_1||M_6||SID_j||PUB_{S_j}||T_5$), $M_{13}$ = $Enc_{PUB_{U_i}}$ ($M_8$,$M_9$,$M_{11}$). All the intercepted messages are either protected through the one-way hash function H(.) or the encryption. Thus, no vulnerability exists which helps an adversary $\mathcal{A}$ to achieve the $UID_i$. In many realtime application, it is expected that an adversary $\mathcal{A}$ must not be able to trace the user and messages communicated by him/her. The $\mathcal{A}$ can trace $U_i$ if and only if an identity of the $U_i$ is revealed. Thus, the proposed LDA-2IoT scheme achieves anonymity and tracebility. 
\subsubsection{Achieves Mutual Authentication and Session Key Agreement}
The mutual authentication property assures each party that the message is received from the valid source. After receiving of the first message from $U_i$, the $GW$ device retrieves it's identity and performs verification of H(H($PUB_{U_i}||UID_i||K_s)||T_1||r_t$) $\stackrel{?}{=} $H($H(PUB_{U_i}||UID_i*||K_s)||T_1||r_t$). Any adversary $\mathcal{A}$ uses $PUB_{U_i}$ to prove himself/herself as a valid user, $\mathcal{A}$ does not get success due to presence of parameters like $X_2$ and $K_s$ in the verification which are not available with $\mathcal{A}$. Similarly, verification H($PUB_{S_j}||T_3||H(SID_j||K_s||PUB_{S_j})$) $\stackrel{?}{=}$ $M_2$ assure about the authenticity of $S_j$ to the $GW$. The verification H($PUB_{S_j}||SID_j||PUB_{GW_k}||T_2$) $\stackrel{?}{=}$ $Temp_1$ helps sensing device $S_j$ to authenticate the $GW$ and the verification H($PUB_{U_i}||UID_i||PUB_{GW_k}||T_6$) $\stackrel{?}{=}$ $M_{11}$ helps $U_i$ to authenticate the $GW$. The computed session key SK = H($M_8||M_9||T_6||X_2||r_1||r_2||r_t$) also includes identities of each entity in indirect manner thus the proposed LDA-2IoT protocol achieves mutual authentication and session key agreement.
\subsubsection{Secure against Replay Attack}
In the replay attack, an adversary $\mathcal{A}$ replays previously communicated messages after some time or in the next session. To provide security against the replay attack, we use random parameters and timestamps in the proposed scheme. Each communicated message contains time-stamp $T_i$ which is validated by the receiving entity through $\Delta$T $\leq$ $T_i*$ - $T_i$ verification where $T_i*$ is the current time at receiver side and $\Delta$T predefined maximum threshold time. Even though $\mathcal{A}$ replays any message, the LDA-2IoT receiver will catch that the received message is replayed. Thus, the proposed LDA-2IoT scheme is secure against the replay attack.
\subsubsection{Secure against User/Sensor Level Side Channel Attack}
In this paper, we propose a LDA-2IoT which reduces numerous user registrations and achieves hierarchical security. The level $l_i$ defines the level of user $U_i$ in the hierarchy and $l_j$ defines the level of sensing device $S_j$ in the deployment. If an adversary $\mathcal{A}$ gets the level $l_i$ then he/she can guess the role of $U_i$ in the organization, similarly if $l_j$ is available to the $\mathcal{A}$ then he/she can guess the sensing device deployment location. Thus, it is important to secure $l_i$ and $l_j$. In the proposed LDA-2IoT, none of the entity (not even $GW$) stores $l_i$ and $l_j$. The $U_i$ stores $l_i$ in parameter $B_i$ = H($l_i || K_s || H(PUB_{U_i}))$ and the $S_j$ stores $l_j$ in parameter $D_j$ = H($l_j || K_s || H(PUB_{RSN_j}))$ which are protected by one-way hash function and the gateway master $K_s$. Hence, the proposed LDA-2IoT scheme is secured against a level side channel attack.   
\subsubsection{Key Establishment with Perfect Forward Secrecy}
In perfect forward secrecy, we assume that the adversary $\mathcal{A}$ somehow obtains the user secret key $RU_i$ and sensing device secret key $RSN_j$, then the adversary $\mathcal{A}$ can retrieve the \{$M_4$,$r_1$\} from the message $M_5$ through the knowledge of $RSN_j$. The hash function protects the message $M_4$, and $r_1$ is an unknown random number that does not provide any useful information. Similarly, through the $RU_i$, an adversary $\mathcal{A}$ can obtain the \{$M_8$,$M_{10}$,$M_{12}$\}. These all the parameters are secured through the one-way hash function and do not provide any useful information. The session key computed using SK = H($M_8||M_9||T_6||X_2||r_1||r_2||r_t$) where parameter $X_2$ is not available to $\mathcal{A}$ and $\mathcal{A}$ can not get $X_2$ or $Y_2$ by just knowing a $RU_i$ and $RSN_j$. We assume that the user device's physical capturing and the user's secret key relieve will not coincide, and this assumption is valid because one is a physical attack while the other is a guessing attack. Thus the proposed LDA-2IoT scheme achieves the perfect forward secrecy.  
\subsubsection{Gateway Device Bypass Attack} In the gateway device bypass attack, an adversary $\mathcal{A}$ tries to behave as a $GW$ or any one of the device $U_i$ or $S_j$ try to behave as a $GW$. In the proposed scheme, during the initialize phase, the $GW$ computes $X_1$ = H($RU_i || UID_i || K_s$), $X_2$ = H($UID_i|| PUB_{U_i}||K_s$) and $B_1$ = H($l_i || K_s || H(PUB_{U_i}))$ for $U_i$ while $Y_1$ = H($RSN_j || SID_j || K_s$), $Y_2$ = H($SID_j || PUB_{S_j} || K_s$), $D_j$ = H($l_j || K_s || H(PUB_{RSN_j}))$ for $S_j$. All these computation involves gateway master secret $K_s$. Thus, neither $\mathcal{A}$ nor the $U_i$ or $S_j$ can compute the above parameters. Hence, the proposed LDA-2IoT scheme is secured against gateway device bypass attack.

\subsubsection{Stolen User Device Attack} In this attack, an adversary $\mathcal{A}$ gets physical user device and retrieves stored parameters \{$X_1$,$X_2$,$B_1$,$RU_i$\}. Now the session key is computed as SK = H($M_8||M_9||T_6||X_2||r_1||r_2||r_t$) where $r_1$, $r_2$ and $r_t$ are the random parameters. If an adversary $\mathcal{A}$ gets the user secret $RU_i$ then also he can not guess random $r_t$. By capturing the user device, an adversary $\mathcal{A}$ can not capture the user identity also. Thus, it is computationally nonfeasible for an adversary to compute the session key $SK$ in polynomial time. 

\subsubsection{Sensing Device Capture Attack} In this attack, an adversary $\mathcal{A}$ gets the physical user device and gets stored parameters \{$Y_1$,$Y_2$,$D_1$,$RSN_j$\}. Now if an adversary $\mathcal{A}$ tries to compute the session key SK = H($M_8||M_{10}||T_6||Y_2||r_1||r_2||r_t$) then it requires three random numbers $r_1$,$r_2$ and $r_t$ as well as the timestamp $T_6$. Thus, even though an adversary $\mathcal{A}$ physically attacks the sensing device as well as track the messages, he can not obtain the $r_1$ and $r_t$ from it. The sensing device does not store the sensing device identity; thus, through the sensing device attack, $\mathcal{A}$ can not track the sensing device also; thus the proposed scheme is secured against the sensing device capture attack.

\subsubsection{User Device Impersonation Attack} In this attack, an adversary $\mathcal{A}$ intercepts all the messages send by user $U_i$ and tries to replace those messages by other manually generated messages. Let $\mathcal{A}$ intercepts \textit{Message 1} =  \{$M_1$,$Temp_0$ $T_1$\}. The $M_1$ is secured through an encryption while $Temp_0$ is secured through the hash function, still let us assume that $\mathcal{A}$ creates \textit{Message 1*} =  \{$M_1*$,$Temp_0*$ $T_1$\} and forwards it to gateway. Now the gateway device $GW$ extracts data from the message $M_1$* and performs the verification H($H(PUB_{U_i}||UID_i*||K_s)||T_1||r_t$) $\stackrel{?}{=}$ $Temp_0$ which contains the fresh random number generated by $U_i$ and the gateway master secret $K_s$ which is computationally unfeasible to generate same $K_s*$ = $K_s$ in a polynomial time for an adversary $\mathcal{A}$.

\begin{table}[H]
    \begin{threeparttable}
     \caption{Security Comparison}
    \centering
\begin{tabular}{|p{11em}|p{0.7em}|p{0.7em}|p{0.7em}|p{0.7em}|p{0.7em}|p{0.7em}|p{0.7em}|p{0.7em}|p{0.7em}|p{0.9em}|p{0.9em}|} \hline
    \textit{\textbf{\tiny{Scheme}}}  & \textit{\textbf{$S_1$}} & \textit{\textbf{$S_2$}} & \textit{\textbf{$S_3$}} & \textit{\textbf{$S_4$}} & \textit{\textbf{$S_5$}} & \textit{\textbf{$S_6$}} & \textit{\textbf{$S_7$}} & \textit{\textbf{$S_8$}} & \textit{\textbf{$S_9$}} & \textit{\textbf{$S_{10}$}} & \textit{\textbf{$S_{11}$}}  \\ \hline
    %\cite{farash2016efficient} & \cmark & \cmark & \xmark & \cmark & \cmark & \xmark & \cmark & \cmark & \cmark & \cmark & \xmark\\ \hline
    %\cite{Das2016} & \cmark & \cmark & \xmark & \cmark & \cmark & \xmark & \cmark & \cmark & \xmark & \cmark & \xmark \\ \hline
    %\cite{wazid2017secure-smarthome} & \xmark & \cmark & \cmark & \xmark & \cmark & \cmark & \cmark & \cmark & \cmark & \cmark & \xmark \\ \hline
    %\cite{Challa2017} & \cmark & \xmark & \cmark & \cmark & \xmark & \cmark & \cmark & \cmark & \xmark & \cmark & \xmark \\ \hline
    \cite{farash2016efficient} & \cmark & \cmark & \cmark & \xmark & \cmark & \cmark & \cmark & \xmark & \cmark & \cmark & \xmark \\ \hline
    \cite{wazid2018} & \cmark & \cmark & \cmark & \xmark & \cmark & \cmark & \xmark & \cmark & \cmark & \xmark & \xmark \\ \hline
    %\cite{AMIN2018} & \cmark & \cmark & \xmark & \cmark & \xmark & \cmark & \cmark & \xmark & \cmark & \xmark & \xmark \\ \hline
    \cite{Zhou2019} & \cmark & \cmark & \cmark & \xmark & \cmark & \cmark & \cmark & \cmark & \xmark & \cmark & \xmark \\ \hline
    \cite{Shin2020} & \cmark & \xmark & \cmark & \cmark & \cmark & \xmark & \cmark & \cmark & \xmark & \cmark & \xmark \\ \hline
    \cite{Jangirala2020} & \cmark & \cmark & \xmark & \cmark & \cmark & \cmark & \xmark & \cmark & \cmark & \xmark & \xmark \\ \hline
    \textbf{\tiny{LDA-2IoT}} & \cmark & \cmark & \cmark & \cmark & \cmark & \cmark & \cmark & \cmark & \cmark & \cmark & \cmark \\ \hline
    \end{tabular}
    \begin{tablenotes}
      \small
      \item \textit{\textbf{Legends:}} $S_1$: Tracebility, $S_2$: Anonymity, $S_3$: Mutual authentication and Integrity, $S_4$: Replay attack , $S_5$: Man-in-The-Middle Attack, $S_6$: Forward secrecy, $S_7$: Gateway by pass attack, $S_8$: Gateway impersonation attack , $S_9$: Sensing device capture attack, $S_{10}$: Privilege insider attack, $S_{11}$: Level Dependent Authentication, \cmark: the protocol supports this feature, \xmark: the protocol doesn't support this feature.
    \end{tablenotes}
    \label{tab:securitycomparison}
    \end{threeparttable}
\end{table}

\subsubsection{Sensing Device Impersonation Attack} In this attack, an adversary $\mathcal{A}$ intercepts all the messages send by the sensing device $S_j$ to the gateway device $GW$. Let $\mathcal{A}$ intercepts \textit{Message 3} = \{$M_3$,$T_3$\} and \textit{Message 5} = \{$M_6$,$M_7$,$T_5$\} and try to replace these messages by \textit{Message 3*} = \{$M_3*$,$T_3$\} and \textit{Message 5*} = \{$M_6*$,$M_7*$,$T_5$\}. The message $M_3$ = $Enc_{PUB_{GW_k}}$($M_2$, $D_j$) is encrypted through the public-key of $GW$. At the other side after receiving this message, $GW$ performs H($PUB_{S_j}||T_3||H(SID_j||K_s||PUB_{S_j})$) $\stackrel{?}{=}$ $M_2$ which includes secure sensor identity $SID_j$ and master secret $K_s$. Thus, it is infeasible to generate a \textit{Message 3*} which is similar to \textit{Message 3}. Thus, the proposed LDA-2IoT scheme is secured against a sensing device impersonation attack.

\subsubsection{Gateway Device Impersonation Attack} In this attack, an adversary $\mathcal{A}$  intercepts all the messages send by the gateway device $GW$ and trie to impersonate as a gateway device. Now let an adversary $\mathcal{A}$ captures \textit{Message 2} = \{$Temp_1$,$T_2$\} , \textit{Message 4} = \{$M_5$,$PUB_{U_i}$,$T_4$\} , \textit{Message 6} = \{$M_{13}$,$T_6$\} and generates new messages \textit{Message 2*} = \{$Temp_1*$,$T_2$\} , \textit{Message 4*} = \{$M_5*$,$PUB_{U_i}$,$T_4$\} , \textit{Message 6*} = \{$M_{13}*$,$T_6$\} and forwards \textit{Message 2*} and \textit{Message 4*} to sensing device $S_j$ and \textit{Message 6*} to user device $U_i$. Now the message $Temp_1$ = H($PUB_{S_j}||H(SID_j||K_s||PUB_{S_j})||PUB_{GW_k}||T_2$) includes the gateway master secret $K_s$ and $SID_j$. The $M_5$ = $Enc_{PUB_{S_j}}$($M_4$, $SID_j$,$r_1$) and $M_{13}$ = $Enc_{PUB_{U_i}}$ ($M_8$,$M_9$,$M_{11},r_1,r_2$) are encrypted by the sensing device secret and the user device secret respectively. Thus it is infeasible to get these both the secrets in polynomial time for an adversary $\mathcal{A}$. Hence, the proposed LDA-2IoT scheme is secured against the gateway device impersonation attack.  

\subsection{Formal Security Proof Using Random Oracle}
\label{Subsec:FormalAnalysis}
\noindent In this section, we perform the formal security analysis for the proposed scheme using a widely accepted and proved secure random oracle based model proposed by Abdalla et al. \cite{Abdalla2005}. The authors in \cite{Abdalla2005} proposed the Real-Or-Random (ROR) model, which helps security designers to prove that the proposed scheme achieves polynomial-time security against an adversary $\mathcal{A}$'s advantage of breaking the security. A similar security model is also used in \cite{Roy2018, Das2018}. 
\begin{enumerate}[a.]
    \item  \textit{Random Oracle:} The random oracle defined as a \textit{H(.)} also called as a hash function which takes message \textit{$m_i$} as a input and computes the one-way irreversible output \textit{$r_i$} \cite{wazid2018, Wazid2018secure}. Whenever an adversary $\mathcal{A}$ generates a challenge with \textit{$m_i$}, random oracle challenger $\mathcal{C}$ computes \textit{$r_i$} = \textit{H($m_i$)} and stores it in the list \textit{L} initialized with NULL  value as a pair of \textit{($m_i$, $r_i$)}. 

     \item \textit{Oracle Participants:} There are three participants in the proposed LDA-2IoT scheme, The user $U_i$, the gateway device $GW$, and the sensing device $S_j$. 

     \item \textit{Oracles:} ${\chi^{p}_{U_i}}$, ${\chi^{q}_{GW}}$, and ${\chi^{r}_{S_j}}$ are oracles with the instances $p$, $q$ and $r$  for the $U_i$, $GW$ and $S_j$ respectively, which are also called as a participants for the protocol $LDA-P$.  

     \item \textit{Oracle Freshness} If using the reveal query $\mathcal{R}$($\chi^{x}$), an adversary $\mathcal{A}$ does not get success in receiving original session key $\mathcal{SK}$ then the oracles,  ${\chi^{p}_{U_i}}$, ${\chi^{q}_{GW}}$, and ${\chi^{r}_{S_j}}$ are considered as a fresh oracles.

     \item \textit{Oracles Partnering:} Oracle instances ${\chi^{x}}$ and ${\chi^{y}}$ are called partner oracles if and only if they fulfill the following criteria simultaneously:
        \begin{itemize}
            \item Both instances ${\chi^{x}}$ and ${\chi^{y}}$ are in the acceptance state.
            \item Both ${\chi^{x}}$ and ${\chi^{y}}$ share the common session id $sid$ and achieve the mutual authentication. "$sid$" is transcript of all the communicated messages between oracles.
            \item Both ${\chi^{x}}$ and ${\chi^{y}}$ satisfy the partner identification and vice-versa.
            \item No instance other than ${\chi^{x}}$ and ${\chi^{y}}$ accept with the partner identification equal to ${\chi^{x}}$ and ${\chi^{y}}$.
        \end{itemize}

    \item \textit{Adversary:} Let us assume that an adversary $\mathcal{A}$ is an eavesdropper who controls the complete communication channel defined over the Dolev-Yao model \cite{Dolev1981}. An adversary $\mathcal{A}$ can read, modify, inject, or fabricate the messages on the communication channel for the proposed network model. An adversary $\mathcal{A}$ has access for the following random oracle queries, which gives numerous capabilities to $\mathcal{A}$ for capturing and modifying the communicated messages and data.
    \begin{enumerate}[1.]
        \item \textit{$\mathcal{R}$ ($\chi^{x}$)} The \textit{Reveal} query $\mathcal{R}$ provides current session key \textit{SK} to the adversary $\mathcal{A}$ which is created by oracle instance $\chi^{x}$ and it's partnering instance.
        \item \textit{$\mathcal{E}$ ($\chi^{x}$,$\chi^{y}$)} The \textit{Execute} query is formed as a passive attack on the communication between oracle participants $\chi^{x}$ and $\chi^{y}$. This query provides all communicated messages to the adversary $\mathcal{A}$.
        \item \textit{$\mathcal{S}$ ($\chi^{x}$,$m_i$)} The \text{Send} query is formed as an active attack performed by $\mathcal{A}$ on instance $\chi^{x}$ where $\chi^{x}$ can receive the message $m_i$ as well as send the reply as a message $m_i$ to $\mathcal{A}$.
        \item \textit{CorruptUserDevice($\chi^{x}$)} The \textit{CorruptUserDevice} query models that the user $U_i$'s  device is available with $\mathcal{A}$ and $\mathcal{A}$ can capture all the data stored in it.
        \item \textit{CorruptSensingDevice ($\chi^{y}$)} The \textit{CorruptSensingDevice} query models that the sensing device $S_j$ is available with $\mathcal{A}$ and $\mathcal{A}$ can capture all the data stored in it using power analysis or reverse engineering attack \cite{Messerges1999, Kocher1999}.
        \item \textit{CorruptUserLevel ($\chi^{x}$)} The \textit{CorruptUserLevel} query models that the level of user $U_i$ is available with an adversary $\mathcal{A}$.
        \item \textit{CorruptSensingLevel ($\chi^{y}$)} The \textit{CorruptSensingLevel} query models that the level of sensing device $S_j$ is available with an adversary $\mathcal{A}$.
        \item \textit{$\mathcal{T}$ ($\chi^{x}$)} Before starting of this oracle game, an unbiased coin b get tossed. The output of this toss decides the return value for the \textit{Test} query $\mathcal{T}$. If the recently generated session key between the user $U_i$ and the sensing device $S_j$ is $SK$ and an adversary $\mathcal{A}$ performs the test query on an instance $\chi^{p}$ which is the instance of $U_i$ or its partner instance $\chi^{r}$ which is an instance of $S_j$ than if the toss output is b = 1 than the participant instance $\chi^{x}$ returns an original session key. In contrast, if the output is b = 0, then the $\chi^{x}$ returns a random value of the session key $SK$'s size to an adversary $\mathcal{A}$. If none of the condition matches, then an instance $\chi^{x}$ returns NULL. The semantic security of the session key is designed based on the \textit{Test} query.
        \end{enumerate}
        
      \item \textit{Session key symmetric security:} The semantic security of the session key $SK$ generated between the user $U_i$ and the sensing device $S_j$ depends on an adversary $\mathcal{A}'s$ capability of indistinguishability between the actual session key and the random number. The output of a test query $\mathcal{T}$ depends on the value of \textit{b'} guessed by an adversary $\mathcal{A}$. If the value of \textit{b'} is similar to the value of \textit{b} which is a hidden bit set by an oracle instance $\chi^{x}$ and used by $\mathcal T(\chi^{x})$ to retrieve the original session key. Overall, the game depends on the correct guess by $\mathcal{A}$ for the bit \textit{b}. If an adversary guesses the correct value of b, then it gets the correct session key. 

    Let $\mathcal{SC}$ define the position in which an adversary gets the success in this game. The advantage of an adversary $\mathcal{A}$ in capturing the correct session key $SK$ for the proposed protocol $LDA_P$ is defined as a \textit{$Adv^{LDA}_p$}. \textit{$Adv^{LDA}_p$} represents the success of an adversary, and if the \textit{$Adv^{LDA}_p$} is negligible, then we can say that the proposed scheme is secured under the ROR model. Thus, we can define \textit{$Adv^{LDA}_p$} as \textit{$Adv^{LDA}_p$}($\mathcal{A}$) = 2*Pr[$\mathcal{SC}$] - 1 which is similar to \textit{$Adv^{LDA}_p$}($\mathcal{A}$) = 2*Pr[b'= b] - 1. Where Pr[$\mathcal{SC}$] represents the probability for the success of an adversary $\mathcal{A}$. If we can prove that the \textit{$Adv^{LDA}_p$} is negligible under the proposed scheme \textit{$LDA_P$}, then we can say that the proposed scheme is secure.
    
    \textbf{Semantic Security for the Password based protocol:} The semantic security for the password based protocol \textit{$LDA-P_{pw}$} defines an adversary $\mathcal{A}$'s capability of guessing the correct password. A password based protocol \textit{$LDA-P_{pw}$} is semantically secure if the advantage function \textit{$Adv_{LDA-P_{pw}}$} is negligible under the condition: \textit{$Adv_{LDA-P_{pw},|\mathcal{D}|}$}($\mathcal{A}$) $\geq$ max($q_s$, ($\frac{1}{|\mathcal{DS}|}$, $\rho_{fp}$)). In this equation, $q_s$ represents the number of send queries($\mathbfcal{S}$), $|\mathcal{DS}|$ shows the finite size of the password dictionary, $\rho_{fp}$ shows probability of the false positive occurrence by an adversary $\mathcal{A}$ \cite{wazid2018, Wazid2018secure}.
\end{enumerate}    
\textit{Security Proof:} We use security model discussed above for to prove formal security of the proposed scheme. Authors in \cite{Wazid2018secure, wazid2018, Das2018} also provided formal proof using random oracle for their schemes.  
\begin{thm}
If $\mathcal{A}$ is a polynomial time attacker running against the proposed protocol $LDA-P$ within a limited time $t$. Let $q_h$ determines the range space of hash ($\mathcal{H}$) queries, $q_s$ denotes the number of send ($\mathcal{S}$) queries, $q_e$ represents the number of execute ($\mathcal{E}$) query, the uniformly distributed password dictionary is defined as \textit{DC} either against the user $U_i$ or the sensing device $S_j$ and $Adv_{\rho}^{ECDLP}$ defines the advantage of $\mathcal{A}$ of breaking the discrete logarithm problem of $\mathcal{A}$ then we can say that the proposed protocol is secured if, 
\begin{equation}
    \label{eq2}    
   \begin{split}
    \textit{$Adv^{LDA}_p$}(\mathcal{A}) \leq  \frac{q^2_{h}}{2^{l_h}} + max(q_s, (\frac{1}{|\mathcal{DC}|}, \rho_{fp})) \\ + Adv_{\rho}^{ECDLP} + (\frac{1}{2^{l_j}})
    \end{split}
\end{equation}
In equation \ref{eq2}, $l_h$ is the size of the return value of a hash ($\mathcal{H}$) query generated by an adversary $\mathcal{A}$ in bits, $l_r$ is the size of the random nonce generated by the protocol $LDA-P$. $|\mathcal{DC}|$ shows the finite size of a password dictionary, and $\rho_{fp}$ shows the probability of a false positive occurrence by $\mathcal{A}$. 
\end{thm}
\begin{proof}
The proposed protocol is secured if the $Adv^{LDA}_p$($\mathcal{A}$) is negligible using the ROR model. Similar proof is also discussed in \cite{Wazid2018secure, wazid2018, Das2018}. We define five games, say \textit{$Gm_0$}, to \textit{$Gm_4$} to prove the security of the proposed scheme. Now, let us define an event \textit{$SC_i$} which represents the correct guess for the bit $b$ in each game \textit{$Gm_i$} via the test query $\mathcal{T}$ by an adversary $\mathcal{A}$. 
\end{proof}

\noindent \textit{$Gm_0$:} The first game $Gm_0$ is the original security game which is corresponding to an original attack performed by an adversary $\mathcal{A}$ on the \textit{LDA-P}. At the beginning of the game, adversary $\mathcal{A}$ chooses bit $b$. Hence it follows that, 
\begin{equation}
    \label{eq3}
    \textit{$Adv^{LDA}_p$}(\mathcal{A}) = 2*Pr[\mathcal{SC} _0] - 1.
\end{equation}

\noindent \textit{$Gm_1$:} The $Gm_1$ is modelled as a passive attack in which $\mathcal{A}$ performs execute query \textit{$\mathcal{E}$(${\chi^{p}_{U_i}}$, ${\chi^{q}_{GW}}$, ${\chi^{r}_{S_j}}$)} and captures all communicated messages (\textit{Message 1} to \textit{Message 7}). Based on all these messages $\mathcal{A}$ tries to determine the session key $SK$ and after completion of the game $\mathcal{A}$ performs a test query $\mathcal{T}$. The output of $\mathcal{T}$ determines weather it is  veritable session key or the random number. The session key is computed by the user $U_i$ and the sensing device $S_j$ as \textbf{SK} = H($M_8||M_9||T_6||X_2||r_1||r_2||r_t$) and \textbf{SK} = H($M_8||M_{10}||T_6||Y_2||r_1||r_2||r_t$) respectively. The session key computation involves $M_8$, $M_9$, $M_{10}$ and random numbers which are secured through the $RU_i$ and $RSN_j$. Since, interception of the messages \textit{Message 1} to \textit{Message 7} does not lead to compromise of the session key $SK$ or any other secret credentials. Thus, the winning probability of the adversary $\mathcal{A}$ does not increase in $Gm_1$.

\begin{equation}
    \label{eq4}
    Pr[\mathcal{SC} _0] =  Pr[\mathcal{SC} _1]. 
\end{equation}

\noindent \textit{$Gm_2$:} The $Gm_2$ involves two more queries in the $Gm_1$. The $Gm_2$ executes \textit{Send} query and Hash \textit{H(.)} through which an adversary $\mathcal{A}$ communicates with the user $U_i$ and the sensor $S_j$. Through the several \textit{H(.)} queries, $\mathcal{A}$ verifies hash digest. Thus, $Gm_2$ is an active attack in which $\mathcal{A}$ tries to convince the $U_i$ and $S_j$ to accept the forged messages. The messages $M_8$ = H($r_3||M_6||M_7||T_6$), $M_9$ = H($SID_j||PUB_{S_j}||K_s$), $M_{10}$ = H($PUB_{U_i}||UID_i||K_s$), $M_{11}$ = H($PUB_{U_i}||UID_i||PUB_{GW_k}||T_6$), $M_{12}$ = H($PUB_{GW_k}||SID_j||PUB_{S_j}||T_6$) involves throughout the use of random numbers, time-stamps, sensing device identity, gateway master secret, user identity which will not provide any success to an adversary $\mathcal{A}$ in collusion verification of the generated message digest. Thus, through the birthday paradox, it follows that,

\begin{equation}
    \label{eq5}
    Pr[\mathcal{SC} _1]  -  Pr[\mathcal{SC} _2] \leq  \frac{q^2_{h}}{2^{l_h}}.
\end{equation}

\noindent \textit{$Gm_3$:} The $Gm_3$ translated from $Gm_2$. The $Gm_3$ performs all the \textit{Corrupt} queries. Through the query \textit{CorruptUserDevice}, an adversary $\mathcal{A}$ receives all the stored parameters like $X_1$, $X_2$, $K_1$, $B_1$ and the other curve parameters. Now, $\mathcal{A}$ tries to guess the correct user ID and password PW for the user $U_i$ through the dictionary attack. To guess the correct password, $\mathcal{A}$ needs $TPW$ and $R_a$ to validate the TPW* = H($UPW_i*||R_a*$). The value of $R_a$ is random value and it's correct guess depends on the correct guess for an identity $UID_i$. Thus, due to these limitations for the \textit{Send} query access in a polynomial time, it is infeasible to guess the correct pair of ($UID_i$, $UPW_i$) in a polynomial time. In similar way, Thus we obtain that, 

\begin{equation}
    \label{eq6}
    Pr[\mathcal{SC} _3] - Pr[\mathcal{SC} _2] \leq max(q_s, (\frac{1}{|\mathcal{DS}|}, \rho_{fp}))
\end{equation}

\noindent \textit{$Gm_4$:} The $Gm_4$ is translated from the $Gm_3$. In this game an adversary $\mathcal{A}$ performs \textit{CorruptUserLevel($\chi^{x}$)}, \textit{CorruptSensingLevel($\chi^{y}$)}, \textit{CorruptSensingDevice}. Through these queries, $\mathcal{A}$ tries to get the level of the user device or the sensor device. Now, let us assume that the probability of guessing the correct level is $\frac{1}{2^{l_j}}$ where $2^{l_j}$ represents the number of bits used for the level. Thus, after guessing the level of user device or sensing device, $\mathcal{A}$ tries to validate it's guess. To validate the user level $l_i$, an adversary $\mathcal{A}$ requires $B_i$ = H($l_i || K_s || UID_i)$ and to validate the sensing device level $l_j$, an adversary $\mathcal{A}$ needs $D_j$ = H($l_j || K_s || SID_j)$. To get these parameters, $\mathcal{A}$ must need the secret key of the user device ($RU_i$) or sensing device ($RSN_j$) which is computationally infeasible for an adversary to get in polynomial time as per Definition \ref{Definition:Def-2}. Thus, we have, 
\begin{equation}
    \label{eq7}
    Pr[\mathcal{SC} _4] - Pr[\mathcal{SC} _3] \leq \frac{1}{2^{l_j}} + Adv_{\rho}^{ECDLP}
\end{equation} 

Now, after completion of all the games, $\mathcal{A}$ doesn't get success. Now $\mathcal{A}$ have only one option left in which $\mathcal{A}$ try to guess the correct value of bit "b" and perform the $\mathcal{T}$ query. The success probability of this query is $\frac{1}{2}$. So after all the games, it is clear that, 
\begin{equation}
    \label{eq8}
    Pr[\mathcal{SC} _4] = \frac{1}{2}
\end{equation} 
Now, from equation \ref{eq3}, we get $\frac{1}{2}$*\textit{$Adv_{LDA-P}$} = \textit{[Pr[$SC_0$] - $\frac{1}{2}$]}. So by using the triangular inequality, we can get the following \textit{[Pr[$\mathcal{SC} _ 1]$} - \textit{[Pr[$\mathcal{SC} _ 4]$} $\leq$ \textit{[Pr[$\mathcal{SC} _ 1]$} - \textit{[Pr[$\mathcal{SC} _ 2]$} + \textit{[Pr[$\mathcal{SC} _ 2]$} - \textit{[Pr[$\mathcal{SC} _ 4]$} $\leq$ \textit{[Pr[$\mathcal{SC} _ 1]$} - \textit{[Pr[$\mathcal{SC} _ 2]$} + \textit{[Pr[$\mathcal{SC} _ 2]$} - \textit{[Pr[$\mathcal{SC} _ 3]$} $\leq$ $\frac{q^2_{h}}{2^{l_h}}$ + max($q_s$, ($\frac{1}{|\mathcal{DC}|}$, $\rho_{fp}$)) + ($Adv_{\rho}^{ECDLP}$) + ($\frac{1}{2^{l_j}}$). Using equations \ref{eq6}-\ref{eq8}, 
\begin{equation}
    \label{eq9}    
   \begin{split}
    |Pr[\mathcal{SC} _ 0]-\frac{1}{2}| \leq  \frac{q^2_{h}}{2^{l_h}} + max(q_s, (\frac{1}{|\mathcal{DC}|}, \rho_{fp})) \\ + Adv_{\rho}^{ECDLP} + (\frac{1}{2^{l_j}})
    \end{split}
\end{equation}
So finally, from the equation \ref{eq3} and \ref{eq9}, we can derive,
\begin{equation}
   \label{eq10}
   \begin{split}
    \textit{$Adv_{LDA-P}$}(\mathcal{A}) \leq  \frac{q^2_{h}}{2^{l_h}}  + max(q_s, (\frac{1}{|\mathcal{DC}|}, \rho_{fp})) \\ + Adv_{\rho}^{ECDLP} + (\frac{1}{2^{l_j}})
    \end{split}
\end{equation}
\section{PERFORMANCE ANALYSIS AND COMPARISON}
\label{Sec:performanceanalysis}
\noindent In this section, we examine and collate the performance of the proposed LDA-2IoT based on communication cost, computation cost, energy consumption, round trip delay and the throughput.  
\subsection{Communication Cost}
\label{Subsec:CommunicationCost}
\noindent The communication cost defines the total number of bits transmitted on the public channel. During the implementation of the proposed LDA-2IoT, we used a python-based programming approach. Table \ref{tab:commcost} shows the total number of bits communicated in the cited schemes over the public channel. The computation of the communication cost is done as follows: to compute the communication cost, we brought the output size for each parameter in the unit of "bits" using python. In our implementation, the size of the generated identity and password is 160 bits. We used SHA-256 as a hash function; thus, the size of the hash output is 256 bits. The timestamp size is 32 bits, and the size of the generated random number is 128 bits. Table \ref{tab:commcost} summarizes communication cost comparison between the proposed scheme and other existing scheme. 
\begin{table}[H]
    \centering
    \caption{Communication Costs Comparison in bits}
\begin{tabular}{|p{13em}|p{3em}|p{4em}|p{3em}|p{4.5em}|} \hline
    \textit{\textbf{Scheme}}  & \textit{\textbf{User}} & \textit{\textbf{Gateway}} & \textit{\textbf{Sensor}} & \textit{\textbf{Total Cost}}  \\ \hline
    \cite{farash2016efficient} & 632 & 792 & 2048 & 3472 \\ 
    %\cite{Das2016} & 512 & 1088 & 384 & 1984 \\ 
    %\cite{wazid2017secure-smarthome} & 480 & 1176 & 512 & 2168 \\ 
    %\cite{Challa2017} & 992 & 1024 & 512 & 2528 \\ 
    %\cite{ALI2018} & 480 & 800 & 480 & 1760 \\ 
    \cite{wazid2018} & 736 & 1344 & 512 & 2592 \\ 
    %\cite{LI2018} & 960 & 960 & 320 & 2240 \\ 
    %\cite{WU2017} & 960 & 1440 & 320 & 2720 \\
    %\cite{AMIN2018} & 832 & 1120 & 320 & 3072 \\ 
    %\cite{poh2019privhome-smarthome} & 640 & 2240 & 800 & 3680 \\ 
    \cite{Zhou2019} & 832 & 2048 & 672 & 3552 \\ \cite{Shin2020} & 1158 & 1560 & 678 & 3552 \\
    \cite{Jangirala2020} & 1012 & 1127 & 517 & 2656 \\
    \textit{LDA-2IoT} & 512 & 1344 & 704 & 2560 \\ \hline
    \end{tabular}
    \label{tab:commcost}
\end{table}
%%%%%%%%%%%%%%%%%%%%%%%%%%%%
%\begin{figure*}
%    \centering
%    \includegraphics[width=\textwidth]{computation.pdf}
%    \caption{Comparison Chart}
%    \label{fig:computation}
%\end{figure*}
%%%%%%%%%%%%%%%%%%%%%%%%%%%%
\subsection{Computation Cost}
\label{Subsec:ComputationCost}
\noindent The computation cost highlights the number of cryptographic operations used in the proposed scheme during the login and authentication stage. It also gives the total time required by those operations at each participant's devices. Let $T_h$, $T_E$/$T_D$,  $T_{P}$ and $T_{fe}$ represent the computation cost of one-way hash function H(.), ECC encryption/decryption operation, ECC Point multiplication and fuzzy extractor respectively. We do not consider the computation of bitwise XOR operation because it takes very little time (almost 0 ms) compare to other operations. Though we implemented the proposed protocol in realtime, we use the self-observations to compute the computation cost for the proposed LDA scheme and other existing schemes. During our implementation, we observe that,
\begin{itemize}
    \item For user device, $T_E$/$T_D$ operation takes 0.07083 seconds, $T_h$ operation takes 0.00041 seconds, the $T_{P}$ operation takes 0.0607 seconds and $T_{fe}$ operation takes 0.0503 seconds.
    \item For sensing device, $T_E$/$T_D$ operation takes 0.08883 seconds, $T_h$ operation takes 0.00084 seconds and the $T_{P}$ operation takes 0.0703 seconds.
    \item For gateway device, $T_E$/$T_D$ operation takes 0.06783 seconds, $T_h$ operation takes 0.00034 seconds and the $T_{P}$ operation takes 0.0589 seconds.
\end{itemize}       
Above all the costs are an average of 100 times verified outputs. Table \ref{tab:comp} summarizes computation cost comparison between the proposed scheme and other existing scheme.  
\begin{table}[H]
    \begin{threeparttable}
    \caption{Computation Costs Comparison}
    \centering
    \begin{tabular}{|p{10em}|p{7em}p{5em}p{5em}|p{2em}|} \hline
    \tiny{\textit{\textbf{Scheme}}}  & \tiny{\textit{\textbf{User}}} & \tiny{\textit{\textbf{Gateway}}} & \tiny{\textit{\textbf{Sensor}}} & \tiny{\textit{\textbf{Time(ms)}}}  \\ \hline
    \tiny{\textbf{\cite{farash2016efficient}}}  & \tiny{\textbf{11*$T_h$}} & \tiny{\textbf{14*$T_h$}} & \tiny{\textbf{7*$T_h$}} & \tiny{\textbf{10.4341}} \\ 
    %\tiny{\textbf{\cite{Das2016}}}  & \tiny{\textbf{$T_{fe}$ + 9*$T_h$}} & \tiny{\textbf{11*$T_h$}} & \tiny{\textbf{5*$T_h$}} & \tiny{\textbf{0.3612}} \\ 
    %\tiny{\textbf{\cite{wazid2017secure-smarthome}}}  & \tiny{\textbf{10*$T_h$ + $T_e$}} & \tiny{\textbf{10*$T_h$ + 2*$T_e$}} & \tiny{\textbf{9*$T_h$ + $T_e$}} & \tiny{\textbf{8.491 ms}}\\ 
    %\tiny{\textbf{\cite{Challa2017}}}  & \tiny{\textbf{$T_{fe}$ + 5*$T_{P}$ + 5*$T_h$}} &\tiny{\textbf{5*$T_{P}$ + 4*$T_h$}} & \tiny{\textbf{4*$T_{P}$ + 3*$T_h$}} & \tiny{\textbf{11.150 ms}}\\ 
    %\tiny{\textbf{\cite{WU2017}}}  & \tiny{\textbf{11*$T_h$}} & \tiny{\textbf{17*$T_h$}} & \tiny{\textbf{6*$T_h$}} & \tiny{\textbf{0.3744}}\\ 
    %\tiny{\textbf{\cite{AMIN2018}}}  & \tiny{\textbf{12*$T_h$}} & \tiny{\textbf{16*$T_h$}} & \tiny{\textbf{6*$T_h$}} & \tiny{\textbf{0.3457}}\\ 
    %\tiny{\textbf{\cite{LI2018}}}  & \tiny{\textbf{2*$T_{P}$ + 8*$T_h$}} & \tiny{\textbf{$T_{P}$ + 9*$T_h$}} & \tiny{\textbf{4*$T_h$}} & \tiny{\textbf{0.4294}} \\ 
    %\tiny{\textbf{\cite{ALI2018}}}  & \tiny{\textbf{2*$T_e$ + 6*$T_h$}} & \tiny{\textbf{5*$T_e$ + 13*$T_h$}} & \tiny{\textbf{1*$T_e$ + 5*$T_h$}} & \tiny{\textbf{8.8 ms}}\\
    \tiny{\textbf{\cite{wazid2018}}}  & \tiny{\textbf{$T_{fe}$ + 13*$T_h$ + 2*$T_e$}} & \tiny{\textbf{5*$T_h$ + 4*$T_e$}} & \tiny{\textbf{4*$T_h$ + 2*$T_e$}} & \tiny{\textbf{8.99}}\\ 
    %\tiny{\textbf{\cite{poh2019privhome-smarthome}}}  & \tiny{\textbf{6*$T_h$}} & \tiny{\textbf{2*$T_h$ + 2*$T_e$}} & \tiny{\textbf{6*$T_h$}} & \tiny{\textbf{0.5030}}\\ 
    \tiny{\textbf{\cite{Zhou2019}}} & \tiny{\textbf{4*$T_{P}$ + 5*$T_h$}} & \tiny{\textbf{3*$T_{P}$ + 7*$T_h$}} & \tiny{\textbf{4*$T_{P}$ + 6*$T_h$}} & \tiny{\textbf{10.693}}\\ 
     \tiny{\textbf{\cite{Shin2020}}} & \tiny{\textbf{3*$T_{P}$ + 14*$T_h$}} & \tiny{\textbf{$T_{P}$ + 12*$T_h$}} & \tiny{\textbf{2*$T_{P}$ + 5*$T_h$}} & \tiny{\textbf{8.66}}\\ 
     \tiny{\textbf{\cite{Jangirala2020}}}  & \tiny{\textbf{5*$T_{P}$ + 13*$T_h$}} & \tiny{\textbf{3*$T_{P}$ + 23*$T_h$}} & \tiny{\textbf{2*$T_{P}$ + 9*$T_h$}} & \tiny{\textbf{22.5}}\\  
    \tiny{\textbf{LDA-2IoT}} & \tiny{\textbf{6*$T_h$ + 2*$T_e$}} & \tiny{\textbf{13*$T_h$ + 6*$T_e$}} & \tiny{\textbf{6*$T_h$ + 3*$T_e$}} & \tiny{\textbf{7.92}}\\ \hline
    \end{tabular}
    \begin{tablenotes}
      \small
      \item \textit{\textbf{Legends:}} $T_h$: One-way hash function cost, $T_E$/$T_D$: ECC Encryption/Decryption cost, $T_{P}$: ECC point multiplication cost
    \end{tablenotes}
    \label{tab:comp}
    \end{threeparttable}
\end{table}

\subsection{Round-trip Delay}
\label{Subsec:RTTAnalysis}
\noindent We computed Round-Trip Delay (RTD) as an average time required by a communicated packet to arrive at the destination from the source \cite{Challa2017}. The round-trip delay involves queuing delay, processing delay, transmission delay, and the propagation delay. The processing delay includes cryptographic operations, while the propagation delay includes travel time required by a packet. For the experimental purpose, through our scenario of numerous users, uni gateway, and multiple sensing devices, we generated simultaneous requests to the gateway device from the user devices for accessing the sensors at different levels. Then, The average RTD at the sensing device, which includes the time between the sensor's reply to the gateway and gateway's response to the sensor, is 0.4825 second. The average RTD at a user device, which includes the time required between sending a request to receiving a reply from the gateway via sensing devices, is 0.5282 seconds. If we send some requests in which the user is not eligible to access the sensor at a particular level, then the RTD gets a little hike due to the gateway node taking little more verification time. If the gateway device does not find a valid user, subsequently, it communicates 0 signal to both the user and sensor to indicate invalid requests received.    
\subsection{Throughput}
\label{Subsec:ThroughputAnalysis}
\noindent We define throughput in either way. The first way is based on the number of bits communicated in unit-time, and the second way is the number of packets transmitted in unit-time. During the implementation of the proposed LDA-2IoT, we gathered data for numerous static users, uni gateway, and numerous static sensors. The throughput is 162 bps, 233 bps, and 91bps at the user, gateway, and sensor. Thus, if we consider the computation cost required for the proposed scheme, then the average transmission time as per the throughput will be 4.28 seconds, 5.16 seconds, and 8.69 seconds needed for the user, gateway, and sensor, respectively. We like to highlight here that we installed "MOSQUITTO" \cite{light2017mosquitto} broker at the gateway for implementation, and we collected the data from the gateway device. Now, if we consider the number of packets transmitted per unit time, then the throughput can be computed as $\frac{total packet received*packet size}{total time}$ \cite{Challa2017}. Thus, by computed using this formula, the average number of MQTT packet received at the user is 7, the sensor is 12 and gateway is 42 where packet size communicated from the user to gateway is 7 byte, gateway to user and sensor is 9 byte and sensor to the gateway is 5 byte through MQTT. Thus, the average throughput for the proposed scheme is 19.48 bps.  
\section{IMPLEMENTATION OF LDA-2IoT}
\label{Sec:Implementation}
\noindent An environment for the implementation of the proposed LDA scheme is highlighted in the following Table \ref{Tab:implementenvt}. We use the laptop and desktop as a user device, the raspberry pi as a gateway device and the NodeMCU connected with the sensors as a sensing device.   
\begin{table}[H]
        \centering
        \caption{Implementation Environment}
        \begin{tabular}{|p{4cm}||p{10cm}|} \hline
        Network Model & Generic IoT Model  \\ \hline
        Broker & Mosquitto \\ \hline
        Protocol  & Using MQTT \\ \hline
        Language & Python \\ \hline
        ECC Curve & NIST P-256 Curve \\ \hline
        Secure channel & By Enabling TLS communication in Mosquitto \\ \hline
        ECC Multiplication & Using double and Add method \\ \hline
        Message format & JSON Type \\ \hline
        User Device & Intel (R) Core (TM) i3-7500 CPU with 2.80 GHz. \\ \hline
        Gateway System & Raspberry Pi 3 Model B, 1 GB RAM. \\ \hline
        Sensing device & NodeMCU + Raspberry Pi \\ \hline
        \end{tabular}
        \label{Tab:implementenvt}
    \end{table}
The following Fig. \ref{Fig:6} shows the computed session key between the user device and the sensing device. 
\begin{figure}[H]
    \centering
    \includegraphics[width=3.2in]{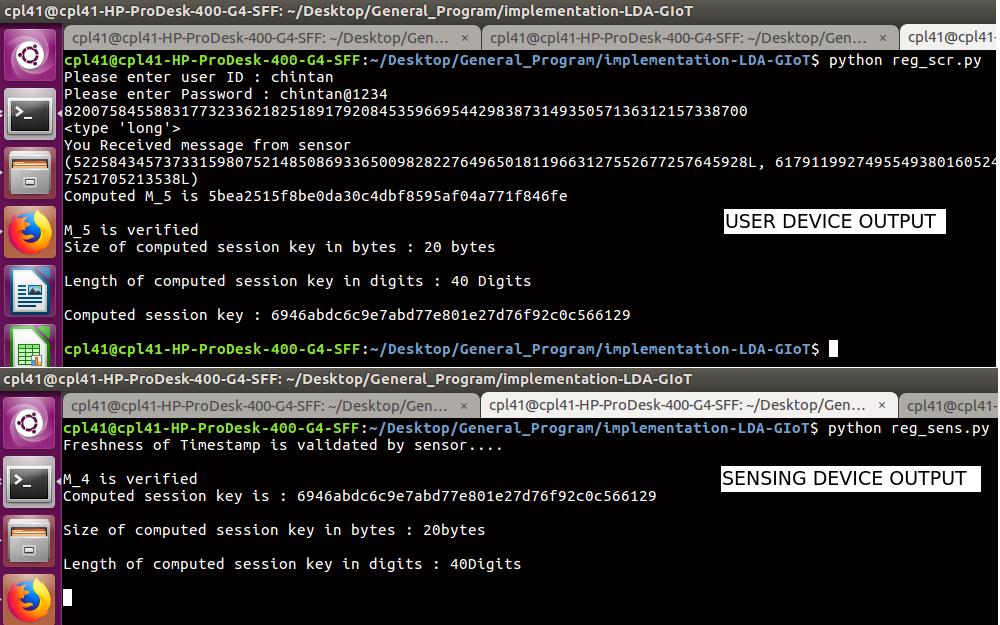}
    \caption{LDA-2IoT Session Key}
    \label{Fig:6}
\end{figure}
%%%%%%%%%%%%%%%%%%%%%%%%%%%%%%%%%%%%%%%%%%%%%%%%
%\begin{figure}[H]
%    \centering
%    \includegraphics[width=3in]{fig5.pdf}
%    \caption{LDA-2IoT Lab Setup}
%    \label{Fig:7}
%\end{figure}
%%%%%%%%%%%%%%%%%%%%%%%%%%%%%%%%%%%%%%%%%%%%%%%%
\section{CONCLUSION}
\label{Sec:Conclusion}
\noindent  In this paper, we introduced a novel IoT authentication approach using an Elliptic Curve Cryptography. We proposed a Level-Dependent Authentication for Generic IoT (LDA-2IoT). The LDA-2IoT reduces the number of user registrations and smooths the access control mechanism of the IoT system. We provided the informal security analysis of the proposed scheme through the Dolev-Yao channel. The formal security analysis of the proposed scheme is given using a widely accepted AVISPA tool and random oracle based ROR Model. The comparative analysis of the LDA-2IoT with the other existing systems shows a little increase in computation and communication costs in the authentication phase. Still, it drastically decreases the efforts in multiple user registration and maintenance of the access control list. The implementation of a proposed LDA-2IoT is done through the MQTT protocol as an application layer protocol and raspberry-pi as a sensing device. Overall the proposed LDA-2IoT opens the new door for the researchers to study access control free, only authentication dependent security systems. The proposed scheme's future work is to perform the feasibility analysis for the proposed LDA approach in the different IoT applications. Another future work of the proposed approach is to implement proposed LDA with other cryptographic approaches like One Time Password (OTP), Physical Unclonable Function (PUF), digital signature, third party certificates, tokenizations and so on.  
\bibliography{reference}

\end{document}